\documentclass[11pt,draftcls,onecolumn]{IEEEtran}

\usepackage{amsmath,amssymb}
\usepackage{graphicx}
\usepackage{subfigure}
\usepackage{bm}
\usepackage{url}
\usepackage{epsfig}
\usepackage{cite}

\newcommand{\E}{\mathbb{E}}

\newcommand{\N}{\mathcal{N}}
\renewcommand{\Re}{\mathbb{R}}
\newcommand{\bx}{x}
\newcommand{\bW}{W}

\newcommand{\eW}{\overline{W}}
\newcommand{\U}{\mathcal{U}}

\newtheorem{thm}{Theorem}

\newtheorem{remark}{Remark}

\graphicspath{{figures/}}

\begin{document}

\title{ Greedy Gossip with Eavesdropping}
\author{\authorblockN{Deniz \"Ustebay$^\star$, Boris Oreshkin, Mark Coates, and Michael Rabbat}\\
\authorblockA{Department of Electrical and Computer Engineering\\ McGill University\\
3480 University St, Montr\'{e}al, Qu\'{e}bec, Canada\\
Email: \url{{deniz.ustebay, boris.oreshkin}@mail.mcgill.ca},\\
\url{{mark.coates, michael.rabbat}@mcgill.ca}}}

\maketitle

\begin{abstract}
This paper presents \emph{greedy gossip with eavesdropping} (GGE), a novel randomized gossip algorithm for distributed computation of the average consensus problem. In gossip algorithms, nodes in the network randomly communicate with their neighbors and exchange information iteratively. The algorithms are simple and decentralized, making them attractive for wireless network applications. In general, gossip algorithms are robust to unreliable wireless conditions and time varying network topologies. In this paper we introduce GGE and demonstrate that greedy updates lead to rapid convergence.  We do not require nodes to have any location information.  Instead, greedy updates are made possible by exploiting the broadcast nature of wireless communications. During the operation of GGE, when a node decides to gossip, instead of choosing one of its neighbors at random, it makes a greedy selection, choosing the node which has the value most different from its own. In order to make this selection, nodes need to know their neighbors' values. Therefore, we assume that all transmissions are wireless broadcasts and nodes keep track of their neighbors' values by eavesdropping on their communications. We show that the convergence of GGE is guaranteed for connected network topologies. We also study the rates of convergence and illustrate, through theoretical bounds and numerical simulations, that GGE consistently outperforms randomized gossip and performs comparably to geographic gossip on moderate-sized random geometric graph topologies.
\end{abstract}


\section{Introduction and Background \protect \footnote{Portions of this work were presented in \cite{ustebay08}, \cite{ustebay08a}, \cite{ustebay09}.}}

Distributed consensus is recognized as a fundamental problem of distributed control and signal processing applications (see, e.g., \cite{tsitsiklis84,SundharRam08a,rabbat05,rabbat06,kashyap07,sundaram08,dimakis06} and references therein). The prototypical example of a consensus problem is computation of the \emph{average consensus}: for a network of $n$ nodes, initially each node has a scalar data value, $y_i$, and the goal is to find a distributed algorithm that asymptotically computes the average, $\bar y=\frac{1}{n}{\sum_{i=1}^n {y_i}}$ at every node $i$. Such an algorithm can further be used for computing linear functions of the data and can be generalized for averaging vectorial data.

One of the algorithms proposed for solving the average consensus problem is distributed averaging \cite{xiao04}. In distributed averaging, every node in the network broadcasts its information at each iteration so that neighboring nodes can receive and use this information for their updates. However, with this scheme the speed of information diffusion across the network  is slow for topologies used to model wireless mesh and sensor networks. The information at each node typically does not change much from iteration to iteration.  Hence, the broadcast medium is not efficiently used. Recently, gossip algorithms have gained attention for the computation of average consensus \cite{boyd06,dimakis06}. In contrast to distributed averaging, gossip algorithms allow only two neighboring nodes to communicate and exchange information at each iteration. Restricting all information exchange to be local in this fashion is attractive from the point of view of simplicity and robustness (e.g., to changing topology and unreliable network conditions).

In this paper we propose a new randomized gossip algorithm, \emph{greedy gossip with eavesdropping} (GGE), for average consensus computation. Unlike previous randomized gossip algorithms, which perform updates completely at random, GGE takes advantage of the broadcast nature of wireless communications and implements a greedy neighbor selection procedure. We assume a broadcast transmission model such that all neighbors within range of a transmitting node receive the message. Thereby, in addition to keeping track of its own value, each node tracks its neighbors values by eavesdropping on their transmissions. At each iteration, the activated node uses this information to greedily choose the neighbor with which it will gossip, selecting the neighbor whose value is most different from its own.  Accelerating convergence in this myopic way does not bias computation and does not rely on geographic location information, which may change in networks of mobile nodes.

Although GGE is a powerful yet simple variation on gossip-style algorithms, analyzing its convergence behavior is non-trivial.  The main reason is that each GGE update depends explicitly on the values at each node (via the greedy decision of with which neighbor to gossip).  Thus, the standard approach to proving convergence to the average consensus solution (i.e., expressing updates in terms of a linear recursion and then imposing properties on this recursion) cannot be applied to guarantee convergence of  GGE. To prove convergence, we demonstrate that GGE updates correspond to iterations of a distributed randomized incremental subgradient optimization algorithm. Similarly, analysis of the convergence rate of GGE requires a different approach than the standard approach of examining the mixing time of a related Markov chain. We develop a bound relating the rate of convergence of GGE to the rate of standard randomized gossip. The bound indicates that GGE always converges faster than randomized gossip, a finding supported by simulation results. We also provide a worst-case bound on the rate of convergence of GGE. For other gossip algorithms the rate of convergence is generally characterized as a function of the second largest eigenvalue of a related stochastic matrix. In the case of GGE, our worst-case bound characterizes the rate of convergence in terms of a constant that is completely determined by the network topology. We investigate the behavior of this constant empirically for random geometric graph topologies and derive lower bounds that provide some characterization of its scaling properties.

\subsection{Background and Related Work}
\label{sec:background}
Randomized gossip was proposed in \cite{boyd06} as a decentralized asynchronous scheme for solving the average consensus problem.  At the $k$th iteration of randomized gossip, a node $s_k$ is chosen uniformly at random.  It chooses a neighbor, $t_k$, randomly, and this pair of nodes ``gossips'': $s_k$ and $t_k$ exchange values and perform the update $x_{s_k}(k) = x_{t_k}(k) = (x_{s_k}(k-1) + x_{t_k}(k-1))/2$, and all other nodes remain unchanged.  One can show that under very mild conditions on the way a random neighbor, $t_k$, is drawn, the values $x_i(k)$ converge to $\bar y$ at every node $i$ as $k \rightarrow \infty$ \cite{xiao04}. Because of the broadcast nature of wireless transmission, other neighbors overhear the messages exchanged between the active pair of nodes, but they do not make use of this information in existing randomized gossip algorithms.

The convergence rate of randomized gossip is characterized by relating the algorithm to a Markov chain \cite{boyd06}. The mixing time of this Markov chain is closely related to the averaging time of the gossip algorithm, and therefore defines the rate of convergence. For certain types of graph topologies, the mixing times are small and convergence of the gossip algorithm is fast. For example, in the case of a complete graph, the algorithm requires $O(n)$ iterations to converge. However topologies such as random geometric graphs~\cite{gupta00} or grids are more realistic for wireless applications. Boyd et al.~\cite{boyd06} prove that for random geometric graphs, randomized gossip requires $O(n^2/\log n)$ transmissions to approximate the average consensus well\footnote{Throughout this paper, when we refer to randomized gossip we specifically mean the natural random walk version of the algorithm, where the node $t_k$ is chosen uniformly from the set of neighbors of $s_k$ at each iteration. For random geometric graph topologies, which are of most interest to us, Boyd et al.~\cite{boyd06} prove that the performance of the natural random walk algorithm scales order-wise identically to that of the optimal choice of transition probabilities, so there is no loss of generality.}.

Motivated by the slow convergence of randomized gossip, Dimakis et al.~introduced \emph{geographic gossip} in~\cite{dimakis06}. Geographic gossip enables information exchange over multiple hops with the assumption that nodes have the knowledge of their geographic locations and the locations of their neighbors. It has been shown that long-range information exchange improves the rate of convergence to $O(n{\sqrt{n/\log n}})$ for random geometric graphs. However, geographic gossip involves overhead due to localization and geographic routing. Furthermore, the network needs to provide reliable two-way transmission over many hops.  Otherwise,  messages which are lost in transit will result in biasing the average consensus computation.

Recently, other fast gossiping algorithms have also been proposed.  Most related is the work of Li and Dai~\cite{li07}, and Jung et al.~\cite{jung07}.  Both approaches are based on directing the exchange of information across the network by constructing \emph{lifted} Markov chains using knowledge of the geographic locations of nodes. As an extension to geographic gossip, Benezit et al.~\cite{benezit07} have recently proposed averaging along paths, an algorithm which converges in $O(n)$ transmissions.  All of these approaches rely on geographic information and thus are not suitable to scenarios where nodes are mobile or location information is not available.  The focus of our work is on providing fast and communication-efficient computation that exploits broadcast transmissions rather than geo-location information to gossip quickly.

Aysal et al.~have proposed {\em broadcast gossip}, a consensus algorithm that also makes use of the broadcast nature of wireless networks \cite{aysal08,aysal08b}.  At each iteration, a node is activated uniformly at random to broadcast its value. All nodes within transmission range of the broadcasting node calculate a weighted average of their own value and the broadcasted value, and they update their local value with this weighted average. Broadcast gossip does not preserve the network average at each iteration.  It achieves a low variance (i.e., rapid convergence), but introduces bias: the value to which broadcast gossip converges can be significantly different from the true average.

Sundhar~Ram et al.~introduced a general class of incremental subgradient algorithms for distributed optimization in \cite{SundharRam08b}. In this study, the effects of stochastic errors (e.g., due to quantization) on the convergence of consensus-like distributed optimization algorithms are investigated. Convergence of their algorithm is guaranteed under certain conditions on the errors, but the convergence rates are not characterized.

Nedi\'c and Ozdaglar have also developed a distributed form of incremental subgradient optimization that generalizes the consensus framework \cite{nedic07}. Our problem formulation is not as general as theirs, but with the specific formulation addressed in this paper we achieve stronger results. In particular, our cost function has a specific form and, by exploiting it, we are able to guarantee convergence to an optimal solution and obtain tight bounds on the rate of convergence as a function of the network topology.

\subsection{Paper Organization}
\label{sec:organization}

The remainder of this paper is organized as follows. Section~\ref{sec:alg} introduces the formal definition of the algorithm. Section~\ref{sec:ggevsrg} presents two bounds; the first relates the performance of GGE to randomized gossip and indicates that GGE always outperforms randomized gossip, and the second is a worst-case upper bound on the rate of convergence of GGE in terms of a topology-dependent constant.  Results from numerical simulations are presented in Section~\ref{sec:sims}.  Motivated by these results we provide a multi-hop extension to our algorithm in Section~\ref{sec:MH}. Section~\ref{sec:summ} summarizes the contributions of the paper and discusses future work.

\section{Greedy Gossip with Eavesdropping (GGE)}
\label{sec:alg}

We consider a network of $n$ nodes and represent network connectivity as a graph, $G = (V,E)$, with vertices $V = \{1,\dots,n\}$, and edge set $E \subset V \times V$ such that $(i,j) \in E$ if and only if nodes $i$ and $j$ directly communicate.  We assume that communication relationships are symmetric and that the graph is connected. Let $\N_i = \{j\ :\ (i,j) \in E\}$ denote the set of neighbors of node $i$ (not including $i$ itself). Each node in the network has an initial value $y_i$, and the goal of the gossip algorithm is to use only local broadcast exchanges to arrive at a state where every node knows the average $\bar y = \frac{1}{n}\sum_{i=1}^n y_i$. To initialize the algorithm, each node sets its gossip value to $x_i(0) = y_i$.

At the $k$th iteration of GGE, a node $s_k$ is chosen uniformly at random from $\{1,\dots,n\}$.  (This can be accomplished using the asynchronous time model described in~\cite{bertsekasPDC}, where each node ``ticks'' according to a Poisson clock with rate 1.) Then, $s_k$ identifies a neighboring node $t_k$ satisfying
\begin{equation}
t_k \in \arg\max_{t \in \N_j} \left\{\frac{1}{2}(x_{s_k}(k-1) -  x_t(k-1))^2\right\}, \label{eqn:ggemax}
\end{equation}
which is to say, $s_k$ identifies a neighbor that currently has the most different value from its own. This choice is possible because each node $i$ maintains not only its own local variable, $x_i(k-1)$, but also a copy of the current values at its neighbors, $x_j(k-1)$, for $j \in \mathcal{N}_i$.
When $s_k$ has multiple neighbors whose values are all equally (and maximally) different from $s_k$'s, it chooses one of these neighbors at random. Then $s_k$ and $t_k$ perform the update
\begin{equation}
  x_{s_k}(k) = x_{t_k}(k) = \frac{1}{2}\big(x_{s_k}(k-1) +  x_{t_k}(k-1)\big), \label{eqn:ggeupdate}
\end{equation}
while all other nodes $i \notin \{s_k, t_k\}$ hold their values at $x_i(k) = x_i(k-1)$.  Finally, the two nodes, $s_k$ and $t_k$, broadcast these new values so that their neighbors have up-to-date information.

Note that one GGE iteration could be accomplished with just two transmissions since $s_k$ already knows the values at its neighbors: one broadcast from $s_k$ to $t_k$ notifying it of the change and simultaneously announcing the new value to all of $s_k$'s neighbors, and one broadcast from $t_k$ to its neighbors to echo the new value to them.  However, in networks with unreliable transmission or systems where nodes periodically shut off their radios to conserve energy, $s_k$ may miss some transmissions from its neighbors, and thus may not always have accurate information about their values.  In this case, mistakes in calculation at $s_k$ in the two-transmission scheme just described would introduce errors, biassing the consensus computation.  To make our algorithm more robust and address the case where $s_k$ does not precisely know the values at all neighbors, we assume a three-transmission version of our scheme throughout the rest of this paper: one transmission from $s_k$ to $t_k$ to initiate gossiping, one from $t_k$ to its neighbors to inform them of the new value, and one from $s_k$ to its neighbors to inform them of its new value.  We comment further on this issue in Section~\ref{subsec:staleInfo}.

\subsection{Initialization}
\label{sec:initial}
Calculating the greedy update in~\eqref{eqn:ggemax} requires nodes to know their neighbors' values. Similar to other randomized gossip algorithms, we assume that at the outset of gossip computation each node $i$ has already discovered its neighbor set, $\mathcal{N}_i$, but it does not know its neighbors' values. Instead, these values are learned during an initialization phase.  During the initialization phase, the node $s_k$ that is activated at iteration $k$ chooses $t_k$ randomly from the subset of its neighbors whose values it does not know, rather than performing a GGE update. Since $s_k$ and $t_k$ broadcast their new values after averaging, the nodes in their neighborhoods overhear and acquire information accordingly. Once $s_k$ has heard from all of its neighbors, the initialization process is complete for that particular node and it chooses $t_k$ greedily via \eqref{eqn:ggemax} for all subsequent iterations.

Figure~\ref{fig:initial} illustrates the effects of this initialization scheme relative to an idealized initialization scheme and a na\"ive initialization scheme.  In the idealized scheme, each node clairvoyantly knows all of its neighbors' initial values, and all nodes immediately commence with greedy updates.  In the na\"ive scheme, before any node begins gossiping, all nodes broadcast once (without performing an update) to inform their neighbors of their starting value.  The results shown correspond to a network of $200$ nodes\footnote{The topology is generated randomly from the family of random geometric graphs on $200$ nodes with the standard connectivity radius.  See Section~\ref{sec:sims} for further details on simulations, and note that the setup used here is the same as that used to generate Figure~\ref{fig:convcompa}.  Although not reported here due to space constraints, we observe similar behavior on other network topologies.}.  Observe that the proposed scheme and the na\"ive scheme attain similar performance, and both involve an overhead of roughly $n=200$ transmissions, which is not substantial relative to the total number of transmission required.  We prefer the proposed initialization scheme to the na\"ive broadcast approach because it does not require any scheduling mechanism.

\begin{figure}[h]
\centering
\includegraphics[height = 0.4\columnwidth,clip=true, viewport=1.3in 3.3in 7in 7.7in]{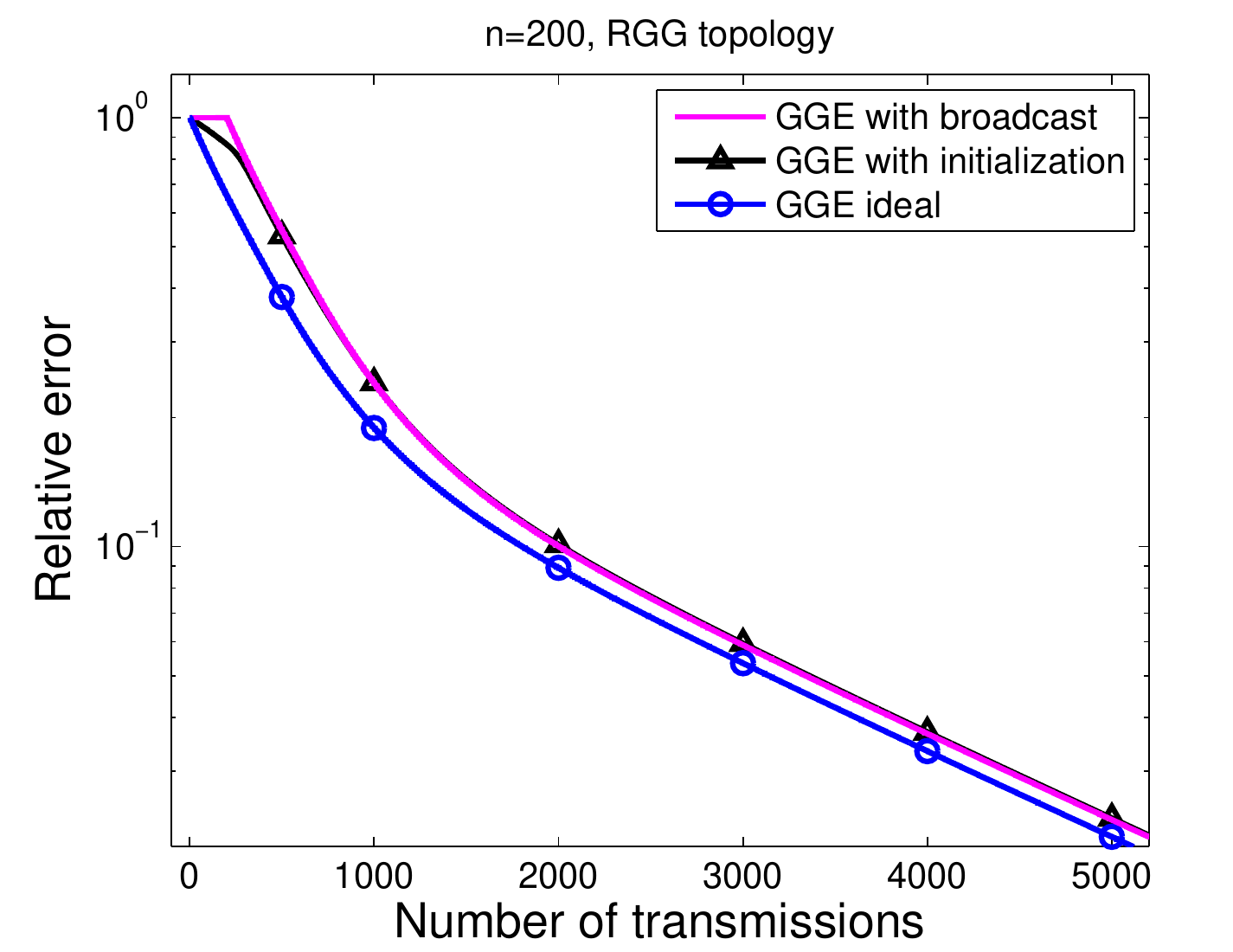}
\caption{A comparison of GGE performance for three different initializations of the algorithm. GGE with broadcast corresponds to the case where each node broadcasts its value once (without performing any updates) to initialize the values at all neighbors, before beginning to gossip.  GGE with initialization corresponds to the proposed scheme, where, before a node has heard from all its neighbors, it gossips randomly with one it has not heard from, and after it has heard from all neighbors its performs greedy updates.  Thus, the proposed GGE initialization scheme gives priority to learning neighbors' values during the initialization phase.  GGE ideal assumes that each node clairvoyantly knows its neighbors' values at the outset.  Note that GGE with broadcast and GGE with initialization give similar asymptotic performance (accruing overhead of roughly $n = 200$ transmissions), although the proposed scheme does not require any scheduling mechanism and during early iterations the proposed scheme actually uses each transmission to perform an update.} \label{fig:initial} %% label for entire figure
\end{figure}

\section{Convergence Analysis}
\label{sec:ggevsrg}

\subsection{Convergence of GGE}
To derive convergence results, we interpret GGE as a randomized incremental subgradient method~\cite{nedic01}. Consider the constrained optimization problem,
\begin{eqnarray*}
  \min_{\bx \in \Re^n} & & \sum_{i=1}^n f_i(\bx)\\
  \mbox{subject to} & & \bx \in X,
\end{eqnarray*}
where we assume that each $f_i(\bx)$ is a convex function, not necessarily differentiable, and $X$ is a non-empty convex subset of $\Re^n$.  An incremental subgradient algorithm for solving this optimization problem is an iterative algorithm of the form:
\begin{equation}
x(k) = \mathcal{P}_X [x(k-1) - \alpha_k g(s_k, x(k-1))], \label{eqn:incrsubgradupdate}
\end{equation}
where $\alpha_k > 0$ is the step-size, $g(s_k, x(k-1))$ is a subgradient\footnote{Subgradients generalize the notion of a gradient for non-smooth functions.  The subgradient of a convex function $f_i$ at $x$ is any vector $g$ that satisfies $f_i(y) \ge f_i(x) + g^T (y - x)$.  The set of subgradients of $f_i$ at $x$ is referred to as the \emph{subdifferential} and is denoted by $\partial f_i(x)$. If $f_i$ is continuous at $x$, then $\partial f_i(x) = \{\nabla f_i(x)\}$; i.e., the only subgradient of $f_i$ at $x$ is the gradient.  A sufficient and necessary condition for $x^*$ to be a minimizer of the convex function $f_i$ is that $0 \in \partial f_i(x^*)$.  See \cite{nedic01} and references therein.} of $f_{s_k}$ at $x(k-1)$, and $\mathcal{P}_X[\cdot]$ projects its argument onto the set $X$.  The algorithm is randomized when the component updated at each iteration, $s_k$, is drawn uniformly at random from the set $\{1,\dots,n\}$, and is independent of $x(k-1)$. Intuitively, the algorithm resembles gradient descent, except that instead of taking a descent step in the direction of the gradient of the cost function, $f(\bx) = \sum_{i=1}^n f_i(\bx)$, at each iteration we focus on a single component, $f_i(\bx)$. The projection, $\mathcal{P}_X[\cdot]$, ensures that each new iterate $x(k)$ is feasible.  Under mild conditions on the sequence of step sizes, $\alpha_k$, and on the regularity of each component function $f_i(\bx)$, Nedi\'{c} and Bertsekas have shown that the randomized incremental subgradient method described above converges to a neighborhood of the global minimizer \cite{nedic01}.

 GGE is a randomized incremental subgradient algorithm for the problem
\begin{eqnarray}
\min_{\bx \in \Re^n} & & \sum_{i=1}^n \max_{j \in
\N_i}\left\{\frac{1}{2}(x_i - x_j)^2\right\} \label{eqn:ggecost}\\
\mbox{subject to} & & \sum_{i=1}^n x_i = \sum_{i=1}^n y_i,\label{eqn:ggeconstraint}
\end{eqnarray}
where $y_i$ is the initial value at node $i$.  The objective function in \eqref{eqn:ggecost} has a minimum value of 0 which is attained when $x_i = x_j$ for all $i,j$.  Thus, any minimizer is a consensus solution.  Moreover, the constraint $\sum_{i=1}^n x_i = \sum_{i=1}^n y_i$ ensures that the unique global minimizer is the average consensus.

To connect the GGE update, \eqref{eqn:ggeupdate}, and the incremental subgradient update, \eqref{eqn:incrsubgradupdate}, observe that $g(k)$, the subgradient of $f_{s_k}(x(k-1)) = \max_{j \in \N_{s_k}} \{\frac{1}{2}(x_{s_k}(k-1) - x_j(k-1))^2\}$, is defined by
\begin{equation}
g_i(k) = \left\{\begin{array}{c l} x_{s_k}(k-1) - x_{t_k}(k-1) & \mbox{
for } i = s_k,\\
-(x_{s_k}(k-1) - x_{t_k}(k-1)) & \mbox{ for } i = t_k,\\
0 & \mbox{ otherwise.}\end{array}\right. \label{eqn:ggesubgrad}
\end{equation}
Here the subscript $i$ indexes the components of the vector $g$. Fixing a constant step size $\alpha_k = \frac{1}{2}$, the update \eqref{eqn:incrsubgradupdate} is identical to \eqref{eqn:ggeupdate}. The recursive update for GGE thus has the form
\begin{equation}
x(k) = x(k-1) - \frac{1}{2} g(k), \label{eqn:ggeiteration}
\end{equation}
Note that the projection is unnecessary since the choice $\alpha_k = \frac{1}{2}$ ensures that the constraint $\sum_{i=1}^n x_i(k) = \sum_{i=1}^n y_i$ is satisfied at each iteration.

 Nedi\'{c} and Bertsekas study the convergence behavior of randomized incremental subgradient algorithms in
\cite{nedic01}. For a constant step size, their analysis only guarantees the iterates $x(k)$ will
reach a neighborhood of the optimal solution: with probability 1, $\min_k f(x(k)) \le \alpha n C^2 / 2$, where
$C \ge \|g(s_k, x(k))\|$ is an upper bound on the norm of the subgradient \cite{nedic01}.  We wish to show that
$x(k)$ actually converges to the average consensus, $\bar x$, the global minimizer of our optimization problem, and not just to a neighborhood of $\bar x$.  By
exploiting the specific form of the GGE cost function, we are able to prove the following stronger result.

\begin{thm}
Let $x(k)$ denote the sequence of iterates produced by GGE.  Then $x(k) \rightarrow \bar x$ almost surely.
\end{thm}
 \begin{proof}
To begin, we examine the improvement in squared error after one GGE iteration. Expanding $x(k+1)$ via the expression \eqref{eqn:ggeiteration}, we have
\begin{eqnarray}
\|x(k) - \bar x\|^2 &=& \|x(k-1) - \frac{1}{2}g(k) - \bar x\|^2\nonumber\\
&=& \|x(k-1) - \bar x\|^2 - \langle x(k-1) - \bar x,\, g(k)\rangle + \frac{1}{4}\|g(k)\|^2. \nonumber
\end{eqnarray}
Based on the definition of $g(k)$ in \eqref{eqn:ggesubgrad}, given $s_k$ and $t_k$, we have
\begin{equation*}
\|g(k)\|^2 = 2(x_{s_k}(k-1) - x_{t_k}(k-1))^2,
\end{equation*}
and,
\begin{eqnarray*}
\langle x(k-1) - \bar x,\, g(k)\rangle &=& \sum_{i=1}^n (x_i(k-1) - \bar x_i) \; g_i(k)\\
&=& (x_{s_k}(k-1) - x_{t_k}(k-1))^2.
\end{eqnarray*}
Therefore, we have \begin{equation}\|x(k) - \bar x\|^2 = \|x(k-1) - \bar x\|^2 - \frac{1}{4}\|g(k)\|^2 \label{eqn:gge_errupdate}\end{equation}
with probability 1, since the expression holds independent of the value of $s_k$ and $t_k$. Recursively applying our update expression, we find that w.p.~1,
\begin{eqnarray}
\|x(k) - \bar x\|^2 &=& \|x(k-1) - \bar x\|^2 -\frac{1}{4}\|g(k)\|^2\nonumber\\
&=& \|x(k-2) - \bar x\|^2 - \frac{1}{4} \sum_{j= k-1}^{k} \|g(j)\|^2\nonumber\\
&\vdots&\nonumber\\
 &=& \|x(0) - \bar x\|^2 - \frac{1}{4} \sum_{j = 1}^{k} \|g(j)\|^2.\nonumber
\end{eqnarray}
Since $\|x(k) - \bar x\|^2 \ge 0$, we must have \begin{equation*}\sum_{j=1}^{k} \|g(j)\|^2 \le 4 \|x(0) - \bar x\|^2\end{equation*} w.p.~1, and, consequently, the series $\sum_{j=1}^{k} \|g(j)\|^2$ converges a.s.~as $k \rightarrow \infty$. Since each term $\|g(j)\|^2 \ge 0$, this also implies that $\|g(k)\|^2 \rightarrow 0$ a.s.~as $k \rightarrow \infty$. However, by definition, $g(k)$ is the subgradient of a convex function, and $g(k)=0$ is both a sufficient and necessary condition for $x(k)$ to be a global minimizer. Thus, $g(k) \rightarrow 0$ a.s.~implies that $x(k) \rightarrow \bar x$ a.s., since $\bar x$ is the unique minimizer of \eqref{eqn:ggecost}-\eqref{eqn:ggeconstraint}.
\end{proof}

\subsection{Convergence Rate: GGE vs. Randomized Gossip}
The following theorem establishes a general expression for the bound on the mean-squared error of GGE after $k$ iterations.  Moreover, it demonstrates that the upper bound on the MSE of GGE is less than or equal to the upper bound on the MSE of randomized gossip.

The GGE updates can also be expressed in the form $\bx(k) = \bW^{GGE}(k) \bx(k-1)$ where $\bW^{GGE}(k)$ is a stochastic matrix with $\bW^{GGE}_{s_k,  s_k}(k) = \bW^{GGE}_{s_k, t_k}(k) = \bW^{GGE}_{t_k, s_k}(k) =
\bW^{GGE}_{t_k, t_k}(k) = \frac{1}{2}$, $\bW^{GGE}_{i,i}(k) = 1$ for all $i \notin \{s_k, t_k\}$, and 0 elsewhere. We denote the application of $k$ successive GGE updates by $\bW^{GGE}(1:k) = \prod_{j=1}^k \bW^{GGE}(j)$.  Likewise, let $\bW^{RG}(1:k) = \prod_{j=1}^k \bW^{RG}(j)$ denote the successive application of $k$ randomized gossip updates.  Let $\eW = \mathbb{E}[\bW^{RG}(k)]$ denote the expected value of the randomized gossip matrix, and let $\lambda_2(\eW)$ denote the second largest eigenvalue of $\eW$.

\begin{thm} \label{th:param}
Let the algorithm input, $x(0)$, be given, let $\bar x$ denote the corresponding average consensus vector.  After $k$ iterations, the expected mean squared error of GGE is upper bounded as follows:
\begin{align}
&\mathbb{E}\left[\| \bW^{GGE}(1:k) x(0) - \bar x\|^2\right]\leq \|x(0) - \bar x\|^2 \prod\limits_{i=1}^k \left(\lambda_2\left(\eW\right) - \xi_i \right) \label{eqn:GGE_bound}
\end{align}
where $\xi_i = 0$ if $\mathbb{E}\left[\|\bW^{GGE}(1:i-1)x(0) - \bar x\|^2\right] = 0$, and otherwise,
\begin{align}
\xi_i &= \frac{\mathbb{E}\left[\sum\limits_{s=1}^n \left(\max\limits_{t\in \N_{s}} \big(x_{s}(i-1) - x_{t}(i-1)\big)^2\right) -\sum\limits_{s=1}^n \left(\frac{1}{|\N_{s}|} \sum\limits_{{t} \in \N_{s}} \big(x_{s}(i-1) -x_{t}(i-1)\big)^2\right)\right]} {2n\,\, \mathbb{E}\left[\|W^{GGE}(1:i-1)x(0) - \bar x\|^2\right]} \ge 0. \label{eqn:xi_k}
\end{align}
Note that $x(i-1)$ is a random quantity determined by the choice of $[s_1~s_2~ ... ~s_{i-1}]$, and the expectation in \eqref{eqn:xi_k} is over these variables.
\end{thm}

\begin{remark}
The analogous expression for randomized gossip is simply~\cite{boyd06}
\begin{equation*}
\mathbb{E}\left[\|\bW^{RG}(1:k)x(0) - \bar x\|^2\right] \le \|x(0) - \bar x\|^2
\lambda_2(\eW)^k.
\end{equation*}
(Note that here, the expectation is taken with respect to both random nodes chosen at each iteration, $[s_1~s_2~ ... ~s_{k}]$ and $[t_1~t_2~ ... ~t_{k}]$, whereas in the expressions in the theorem, the only randomness is in $[s_1~s_2~ ... ~s_{k}]$.)  Since $\xi_i \ge 0$ for all $i=1,\dots,k$, this implies that the upper bound on GGE is uniformly upper bounded by the upper bound for randomized gossip, for any $k \ge 0$ and any input $x(0)$. The upper bound for random gossip is tight; if $v_2$ denotes the eigenvector corresponding to the second-largest eigenvalue of $\eW$, then if $x(0)=c v_2$ for some constant $c$, the upper bound holds with equality (in expectation).
\end{remark}

\begin{remark} The form of the terms $\xi_i$ also provides insight into which scenarios are less favorable for GGE.  In general, we know that randomized gossip is slow to converge on random geometric graphs \cite{boyd06}, and so we hope that $\xi_i > 0$ so that GGE achieves some improvement.  Note that the numerator of $\xi_i$ measures how much larger (on average) a GGE step from $x(i-1)$ is in comparison to the step taken by randomized gossip from the same location. There are two scenarios where the expression for $\xi_i$ in \eqref{eqn:xi_k} evaluates to 0.  The first is when $x(i-1) = \bar x$, in which case consensus has already been achieved.  The second is when the difference between any two neighbors is constant across the network; i.e., $(x_s - x_t)^2 = c$ for all $t \in \N_s$ and all $s=1,\dots,n$.  In this setting, being greedy does not provide any gain, since gossiping with any neighbor provides the same amount of immediate improvement. \end{remark}

\begin{proof}[Proof of Theorem~\ref{th:param}]
We recall the known convergence rate bounds for randomized gossip~\cite{boyd06}, \begin{equation} \label{eqn:RG_bound}
\mathbb{E}\left[||\bW^{RG}(1:k)x(0) - \bar x||^2\right] \leq \lambda_2\left(\eW\right)^k ||x(0) - \bar x||^2
\end{equation}
and the related recursive relationship,
\begin{align}
&\mathbb{E} \left[|| \bW^{RG}(1:k)x(0) - \bar x||^2\right]\nonumber\\
&= \mathbb{E}\left[|| \bW^{RG}(1:k-1)x(0) - \bar x||^2 -  \frac{1}{2n} \sum\limits_{s=1}^n \frac{1}{|\mathcal{N}_{s}|} \sum\limits_{t\in \mathcal{N}_{s}} \big(x_{s}(k-1) - x_{t}(k-1)\big)^2\right] \nonumber\\
&\leq \lambda_2\left(\eW\right) \mathbb{E} \left[|| \bW^{RG}(1:k-1)x(0) - \bar x||^2\right]. \label{eqn:oneRGstep}
\end{align}
We can identify an equivalent relationship derived from applying $k-1$ steps of GGE followed by one step of random gossip:
\begin{align}
&\mathbb{E} \left[||\bW^{RG}(k)\bW^{GGE}(1:k-1)x(0) - \bar x||^2\right] \nonumber\\
&= \mathbb{E} \left[|| \bW^{GGE}(1:k-1)x(0)  -\bar x||^2 -  \frac{1}{2n} \sum\limits_{s=1}^n \frac{1}{|\mathcal{N}_{s}|} \sum\limits_{t\in \mathcal{N}_{s}} \big(x_{s}(k-1) - x_{t}(k-1)\big)^2 \right]\nonumber\\
&\leq \lambda_2\left(\eW\right) \mathbb{E} \left[|| \bW^{GGE}(1:k-1)x(0)- \bar x||^2\right].
\end{align}
With this relationship in hand, we can bound the error of the GGE algorithm by adding and subtracting the effects of making the $k$-th step a randomized gossip update:
\begin{align}
&\mathbb{E} \left[|| \bW^{GGE}(1:k)x(0) - \bar x)||^2\right]\nonumber \\
&~~~= \mathbb{E} \left[|| \bW^{GGE}(1:k-1)x(0) - \bar x )||^2-  \frac{1}{2n}\sum\limits_{s=1}^n \frac{1}{|\mathcal{N}_{s}|} \sum\limits_{t\in \mathcal{N}_{s}} \big(x_{s}(k-1) - x_{t}(k-1)\big)^2\right]  \nonumber\\
&~~- \mathbb{E} \left[ \frac{1}{2n} \sum\limits_{s=1}^n \max\limits_{t\in \mathcal{N}_{s}} \big( x_{s}(k-1) -x_{t}(k-1)\big)^2 +  \frac{1}{2n} \sum\limits_{s=1}^n \frac{1}{|\mathcal{N}_{s}|} \sum\limits_{t\in \mathcal{N}_{s}} \big(x_{s}(k-1) - x_{t}(k-1)\big)^2\right]\nonumber\\
&~~~\leq \left[\lambda_2\left(\eW\right) - \xi_k \right] \mathbb{E}\left[||\bW^{GGE}(1:k-1)x(0) - \bar x||^2\right].
\end{align}
Repeated application of this inequality from $i=1,\dots,k$ yields the bound (\ref{eqn:GGE_bound}).
\end{proof}

\subsection{GGE Convergence Rate: Worst Case Bounds}
\label{sec:bound}

The previous subsection related the performance of GGE to that of standard randomized gossip.  Next, we seek a more direct characterization of the GGE rate of convergence in terms of properties of the underlying communication topology.  We then revisit our comparison to randomized gossip. The rate of convergence for gossip algorithms is typically quantified in terms
of the $\epsilon$-averaging time,
\begin{equation*}
T_{ave}(\epsilon) = \sup_{x(0) \ne 0} \inf \left\{ k\ :\
\Pr\big(\frac{\|x(k) - \bar{x}\|}{\|x(0) - \bar{x}\|} \ge \epsilon
\big) \le \epsilon\right\}.
\end{equation*}
Other gossip algorithms such as randomized gossip and geographic gossip are easily related to a homogeneous Markov chain, and $T_{ave}(\epsilon)$ can be shown to scale as a function of the second largest eigenvalue of $\eW$.  In particular (see Theorem~3 in \cite{boyd06}), $T_{ave}(\epsilon) \le \frac{3 \log \epsilon^{-1}}{\log \lambda_2(\eW)^{-1}}$.  For randomized gossip, the matrix $\eW$ depends on the choice of probabilities assigned to each edge in the network and hence depends on the network topology.

Since, in each iteration of GGE, the greedy decision depends on the gossip values at each node, $x(k)$, our algorithm cannot be related to a homogeneous Markov chain ($t_k$ depends on $x(k)$).  Consequently, the same machinery cannot be used to characterize the rate of convergence for GGE.  The goal of this section is to bound the rate of convergence of GGE through alternative means.  To this end, our main result is the following.

\begin{thm} \label{prop:worstCase}
Let $G = (V,E)$ denote the graph on which we are gossiping, let $x(k)$ denote the vector of GGE values after $k$ iterations, and let $\bar{x}$ denote the average vector. Then
\begin{eqnarray*}
\mathbb{E}\left[\|x(k) - \bar{x}\|^2\right] &\le& A(G)^k \|x(0) - \bar{x}\|^2,
\end{eqnarray*}
where $A(G)$ is the graph-dependent constant defined as
\begin{eqnarray*}
A(G) &=& \max_{x \ne \bar{x}} \frac{1}{n} \sum_{s =1}^n \left(1 -
\frac{\|g_s(x)\|^2}{4 \|x - \bar{x}\|^2}\right),
\end{eqnarray*}
where $g_s(x)$ refers to a subgradient of $f_s(x)$, when viewing GGE as an incremental subgradient algorithm\footnote{We explicitly note that this constant is a function of the underlying topology by writing $A(G)$. The constant $A(G)$ is completely determined by the neighborhood structure of the network because the maximization is over all $x$. For a fixed $x$, the subgradients are determined by the neighborhood structure.}.  Moreover, the $\epsilon$-averaging time for GGE is bounded above by
\begin{eqnarray*}
T_{ave}(\epsilon) &\le& \frac{3 \log \epsilon^{-1}}{\log A(G)^{-1}}.
\end{eqnarray*}
\end{thm}

\begin{remark} Note that the constant $A(G)$ only depends on the topology of the graph. This constant plays a role for GGE similar to that played by the second-largest eigenvalue of $\eW$ for regular gossip algorithms.
\end{remark}

\begin{proof}[Proof of Theorem~\ref{prop:worstCase}] The proof of the first part of Theorem \ref{prop:worstCase} is based on an approach introduced in \cite{burnashev74} and developed in \cite{castro07} for analyzing data-adaptive algorithms.  We begin by recalling the recursion for the mean squared error of GGE after $k$ iterations expressed in~\eqref{eqn:gge_errupdate}:
\begin{eqnarray*}
\|x(k) - \bar{x}\|^2 &=& \|x(k-1) - \bar{x}\|^2 -\frac{1}{4}\|g(k)\|^2 \\
&=& \left(1 - \frac{\|g(k)\|^2}{4 \|x(k-1) - \bar{x}\|^2}\right)\|x(k-1) - \bar{x}\|^2,
\end{eqnarray*}
where $g(k)$ denotes the subgradient at iteration $k$ (when viewing GGE as a randomized incremental subgradient algorithm), and is a random quantity, depending on which node $s(k)$ is activated at iteration $k$.  Let $M(k) = \|x(k) - \bar{x}\|^2$ denote the error after $k$ iterations, and let $N(k) = 1 - \frac{\|g(k)\|^2}{4 \|x(k-1) - \bar{x}\|^2}$ denote the amount of contraction at iteration $k$.  Using these definitions and some successive conditioning, we get
\begin{eqnarray*}
\mathbb{E}[M(k)] &=& \mathbb{E}[N(k)M(k-1)] \\
&=& \mathbb{E}[\mathbb{E}[N(k) M(k-1) | x(k-1)]]\\
&=& \mathbb{E}[M(k-1) \mathbb{E}[N(k) | x(k-1)]]\\
&\vdots& \\
&=& M(0) \mathbb{E}[\mathbb{E}[N(1) | x(0)] \cdots \mathbb{E}[N(k) |
x(k-1)]].
\end{eqnarray*}
Note that $A(G)$ is defined in such a way that $\mathbb{E}\left[N(k) | x(k-1)\right] \le A(G)$ for all $k$.  Therefore, it follows that \begin{eqnarray*} \mathbb{E}[\|x(k) - \bar{x}\|^2] &\le& A(G)^k \|x(0) - \bar{x}\|^2. \end{eqnarray*} Next, we prove the second part of the claim: the bound on $\epsilon$-averaging time.  To do this, we will use the bound we have just derived to develop an upper bound on $\Pr(\|x(k) - \bar{x}\| \ge \epsilon \|x(0) - \bar{x}\|)$, the probability that after $k$ iterations we are still more than a factor of $\epsilon$ away from the initial error.  Applying Markov's inequality and the bound we just derived for $\mathbb{E}[\|x(k) - \bar{x}\|^2]$, we have
\begin{eqnarray*}
\Pr\left(\|x(k) - \bar{x}\| \ge \epsilon \|x(0) - \bar{x}\|\right) &=& \Pr\left(\|x(k) - \bar{x}\|^2 \ge \epsilon^2 \|x(0) - \bar{x}\|^2\right) \\
&\le~& \frac{\mathbb{E}[\|x(k) - \bar{x}\|^2]}{\epsilon^2 \|x(0) -\bar{x}\|^2}\\
&\le~& \epsilon^{-2} A(G)^k.
\end{eqnarray*}
 To get an upper bound on $T_{ave}(\epsilon)$, first note that $\Pr\left(\|x(k) - \bar{x}\| \ge \epsilon \|x(0) - \bar{x}\|\right) \le \epsilon$ provided that $k \ge \frac{3 \log \epsilon^{-1}}{\log A(G)^{-1}}$.  Since in the first part of our proposition, the bound on $\mathbb{E}[\|x(k) - \bar{x}\|^2]$ is based on a worst-case one-step analysis, it is an upper bound on the mean squared error at iteration k, effectively a lower bound on the rate of convergence.  Therefore, we have an upper bound on the $\epsilon$-averaging time for GGE; that is $T_{ave}(\epsilon) \le \frac{3 \log \epsilon^{-1}}{\log A(G)^{-1}}$. \end{proof}

Theorem~\ref{prop:worstCase} provides a direct link between the rate of convergence of GGE and the underlying network topology through the constant, $A(G)$.  This motivates further study of how $A(G)$ behaves for different classes and sizes of network topologies.  Next, we derive a lower-bound on $A(G)$ as a function of the maximum degree of the network, $d_{\max} = \max_i |\N_i|$.  Then we apply this result to characterize $A(G)$ for two-dimensional grid and random geometric graph topologies.

\begin{thm} \label{prop:MSE_LB}
Let $G$ be a graph with $n$ nodes and maximum degree $d_{\max}$.  As above, let $\eW$ denote the expected update for one step of randomized gossip on $G$.  There exists a vector $x \in \Re^n$ with corresponding average consensus vector $\bar x$ such that
\begin{equation}
\frac{\E\|W^{GGE} x - \bar x\|^2}{\|x - \bar x\|^2} \geq
(1-d_{\max}(1-\lambda_2\left(\eW\right))),
\end{equation}
and this implies a lower-bound for $A(G)$,
\begin{equation}
A(G) \geq 1-d_{\max}(1-\lambda_2\left(\eW\right)).
\end{equation}
\end{thm}

The proof appears below.  We can use this result to relate the upper bounds on averaging time for GGE and standard randomized gossip.  Let $\U^{GGE}(G,\epsilon) = \frac{3 \log \epsilon^{-1}}{\log A(G)^{-1}}$ denote the upper bound on the averaging time of GGE obtained in Theorem~\ref{prop:worstCase}, and let $\U^{RG}(G,\epsilon) = \frac{3 \log \epsilon^{-1}}{\log \lambda_2(\eW)^{-1}}$ denote the corresponding upper bound on the averaging time of randomized gossip \cite{boyd06}.  Using two inequalities for the logarithm --- namely, $a < -\log(1-a)$ for $a \in (0,1)$, and $\log \lambda \le \lambda - 1$ for $\lambda > 0$ --- we obtain
\begin{equation}
\U^{GGE}(G,\epsilon) \ge \frac{3 \log \epsilon^{-1}}{d_{\max} (1 - \lambda_2(\eW))} \ge \U^{RG}(G,\epsilon) / d_{\max}.
\end{equation}
In words, the upper bound on the averaging time of GGE is at most a factor of $d_{\max}$ better than the upper bound for randomized gossip.  Of course, this only links the upper bounds of the two algorithms and does not directly relate their actual performance.  However, simulation results presented in the next section indicate that this relationship indeed captures the improvements seen for GGE over randomized gossip.  Moreover, recall that the bounds on the expected improvement after a single gossip iteration are tight for both GGE and randomized gossip.  The bound for GGE in Theorem~\ref{prop:worstCase} becomes an equality for $k=1$ when $x(0)$ is taken to be the $x$ that solves the optimization problem defining $A(G)$.  Similarly, the bound for randomized gossip becomes an equality when $x(0)$ is taken to be the eigenvector corresponding to the second largest eigenvalue of $\eW$.

We are interested in understanding the performance of GGE for applications primarily in wireless networks.  Random geometric graphs, first introduced in \cite{gupta00}, are commonly used to model connectivity in wireless networks for the purpose of analyzing scaling behaviour of algorithms.  A random geometric graph on $n$ nodes is obtained by uniformly assigning each node i.i.d.~coordinates in the unit square and then connecting nodes whose distance is less than connectivity radius $r(n)$.  In this paper we adopt the common scaling $r(n) = \sqrt{\frac{2 \log n}{n}}$, which guarantees the network is connected with high probability \cite{gupta00}.

For a random geometric graph with $n$ nodes, it is known~\cite{boyd06} that, for $r(n)$ as given above, every node has $2\pi \log(n) (1 - o(1))$ neighbors with high probability.  Thus, with $d_{\max} = 2\pi \log(n) (1 - o(1))$, we see that GGE gives essentially a factor of $\log n$ improvement in averaging time over randomized gossip.  For a two-dimensional grid $d_{\max} = 4$, and GGE gives only a constant improvement in averaging time.  These results are illustrated via simulation in the next section.

\begin{proof}[Proof of Theorem~\ref{prop:MSE_LB}]
Our starting point is \eqref{eqn:xi_k} from Theorem~\ref{th:param}.  Since we are focusing on the effect of applying a single gossip iteration, we drop the time index $i$ in \eqref{eqn:xi_k} to simplify the notation:
\begin{align}
\xi &=\frac{ \mathbb{E}\left[\sum\limits_{s=1}^n \left(\max\limits_{t \in \N_s} \big(x_s - x_t\big)^2\right) -\sum\limits_{s=1}^n \left(\frac{1}{|\N_s|} \sum\limits_{t \in \N_s} \big(x_s -x_t\big)^2\right)\right]} {2n\,\, \mathbb{E}\left[\|x - \bar x\|^2\right]} \ge 0, \label{eqn:xi}
\end{align}
Using the fact that the maximum of a set of non-negative values is always less than or equal to the sum of those values, we can write
\begin{eqnarray}
\mathbb{E}\left[\max_{t \in \N_s} (x_s - x_t)^2\right] &\le& \frac{|\N_s|}{|\N_s|}~\mathbb{E}\left[\sum_{t\in \N_s} (x_s - x_t)^2\right] \\
&\le& \frac{d_{\max}}{|\N_s|} ~\mathbb{E}\left[\sum_{t\in \N_s} (x_s - x_t)^2\right].
\end{eqnarray}
Therefore we can upper bound $\xi$ by
\begin{equation}
\xi \leq \frac{(d_{\max}-1)~\mathbb{E}\left[\sum\limits_{s=1}^n \left(\frac{1}{|\N_s|} \sum\limits_{t \in \N_s} \big(x_s -x_t\big)^2\right)\right]} {2n\,\, \mathbb{E}[\|x - \bar x\|^2]}. \label{pro:upperxi}
\end{equation}
Next, take $x$ to be the eigenvector corresponding to the second largest eigenvalue of $\eW$ and equate $W^{RG}(1:k-1)x(0) = x$ in \eqref{eqn:oneRGstep} to get
\begin{eqnarray}
\mathbb{E}\left[{-}\frac{1}{2n} \sum_{s=1}^n \frac{1}{|\N_s|} \sum_{t \in \N_s} (x_s - x_t)^2\right] &\le& - (1-\lambda_2(\eW)) \E[\|x - \bar x\|^2].
\end{eqnarray}
Applying this inequality in \eqref{pro:upperxi} gives
\begin{equation}
\xi \leq (d_{\max}-1)(1-\lambda_2(\eW)).  \label{eqn:xidmax}
\end{equation}

Observe that the one-step bound obtained by taking $k=1$ in~\eqref{eqn:GGE_bound} is tight. In particular, for our choice of $x$ as the eigenvector corresponding to the
second largest eigenvalue of $\eW$, we have equality in \eqref{eqn:GGE_bound}:
\begin{align}
\mathbb{E}\left[\| W^{GGE} x - \bar x\|^2\right] = \|x - \bar x\|^2
\left(\lambda_2\left(\eW\right) - \xi \right). \label{eqn:GGE_boundeq}
\end{align}
Inserting (\ref{eqn:xidmax}) into (\ref{eqn:GGE_boundeq}) leads to the
first claim in the proposition.  Then, combining this bound with the first inequality in Theorem~\ref{prop:worstCase} for the case $k=1$ and $x(0) = x$ yields the desired lower bound on $A(G)$.
\end{proof}

\section{Numerical Simulations}
\label{sec:sims}

In this section we report the results of simulations conducted to compare the performance of GGE with randomized gossip~\cite{boyd06} and geographic gossip~\cite{dimakis06} for a variety of initial conditions. We also compare the empirically achieved convergence rates to the bound established in Section~\ref{sec:bound} and investigate how this bound behaves as the number of nodes in the network grows.

\subsection{Comparison of Convergence Rates}

We first compare the convergence rates of GGE with randomized gossip and geographic gossip by examining the reduction they achieve in relative error,~$||x(k)-\bar{x}|| \over ||x(0)-\bar{x}|| $, as a function of the number of transmissions (communication complexity). Since the number of transmissions per iteration is different for each algorithm, this is a fairer comparison than examining convergence rate relative to the number of iterations.  Randomized gossip requires two wireless transmissions per iteration,  GGE requires three transmissions (see the discussion in Section~\ref{sec:alg}), and geographic gossip has a variable number of transmissions per iteration, depending on the number of hops between the gossiping nodes. We simulate networks with random geometric graph topologies, and all figures show averages over 100 realizations of the random geometric graph. We examine performance for four different initial conditions, $x(0)$, in order to explore the impact of the initial values on performance. The first two of these cases are a Gaussian bumps field, and a linearly-varying field. For these two cases, the initial value $x(0)$ is determined by sampling these fields at the locations of the nodes. The remaining two initializations consist of the ``spike'' signal, constructed by setting the value of one random node to 1 and all other node values to 0, and a random initialization where each value is i.i.d. drawn from a Gaussian distribution $\mathcal{N}(0,1)$ of zero mean and unit variance. The first three of these signals were also used to examine the performance of geographic gossip in~\cite{dimakis06}.

Figs.~\ref{fig:convcomp}(a)-(d) show that GGE converges to the average at a much faster rate (both initially and asymptotically) than randomized gossip for all initial conditions. The initial rate of convergence of GGE is faster than geographic gossip for all but the linearly-varying field, and for the simulated network size ($n=200$), the asymptotic rates of reduction in relative error are similar for the two algorithms. Out of these candidate initializations, the linearly-varying field is the worst case. This is not surprising since the convergence analysis conducted in Section~\ref{sec:ggevsrg} suggests that constant differences between neighbors causes GGE to provide minimal gain.

\begin{figure*}[tbp]
\centering
\subfigure[Gaussian bumps convergence rate comparison]{
\label{fig:convcompa} %% label for second subfigure
\includegraphics[height = 0.35\columnwidth,clip=true, viewport=1.3in 3.3in 7in 7.7in]{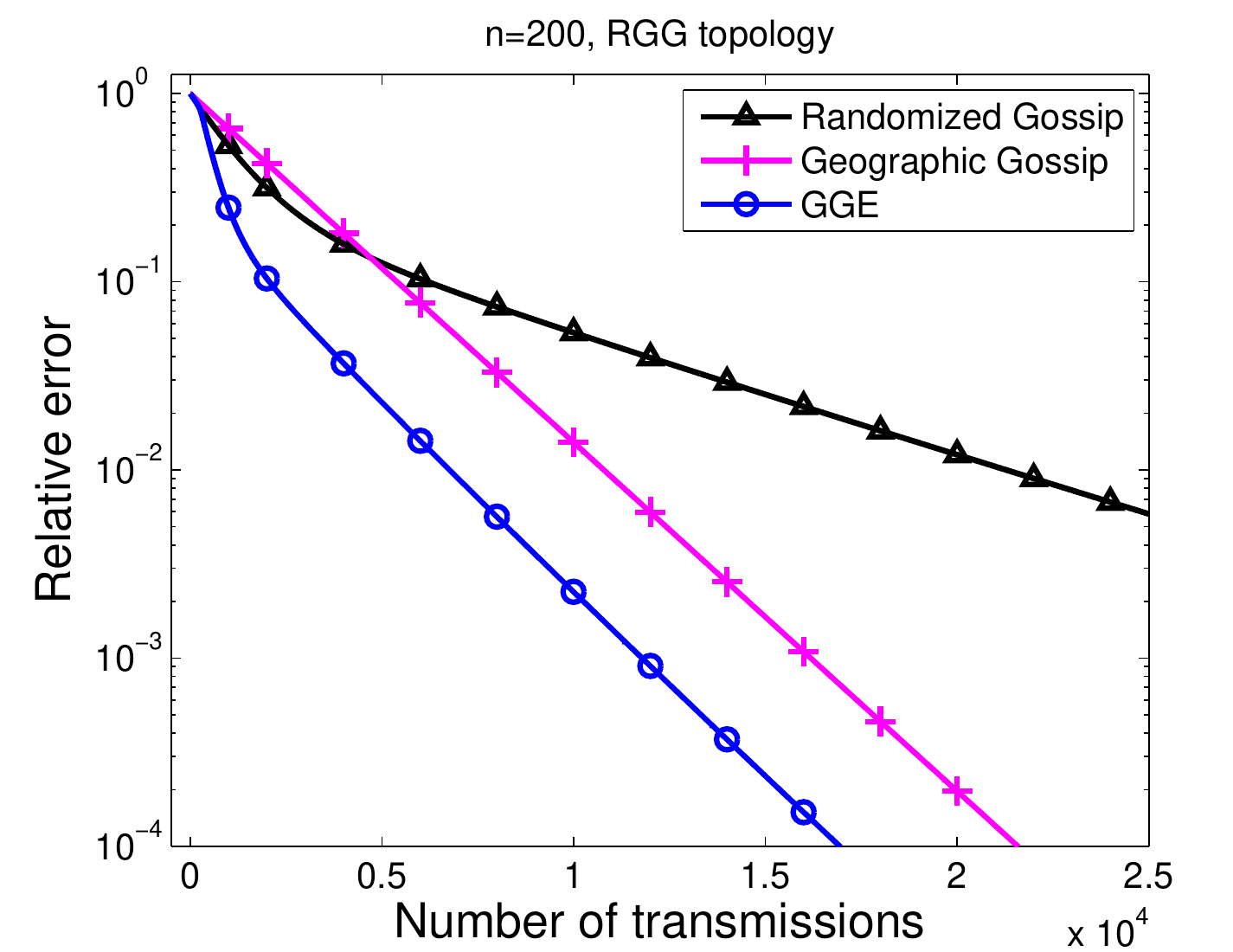}}
\hfill \subfigure[Linearly-varying field convergence rate comparison]{
\label{fig:convcompb} %% label for first subfigure
\includegraphics[height = 0.35\columnwidth,clip=true, viewport=1.3in 3.3in 7in 7.7in]{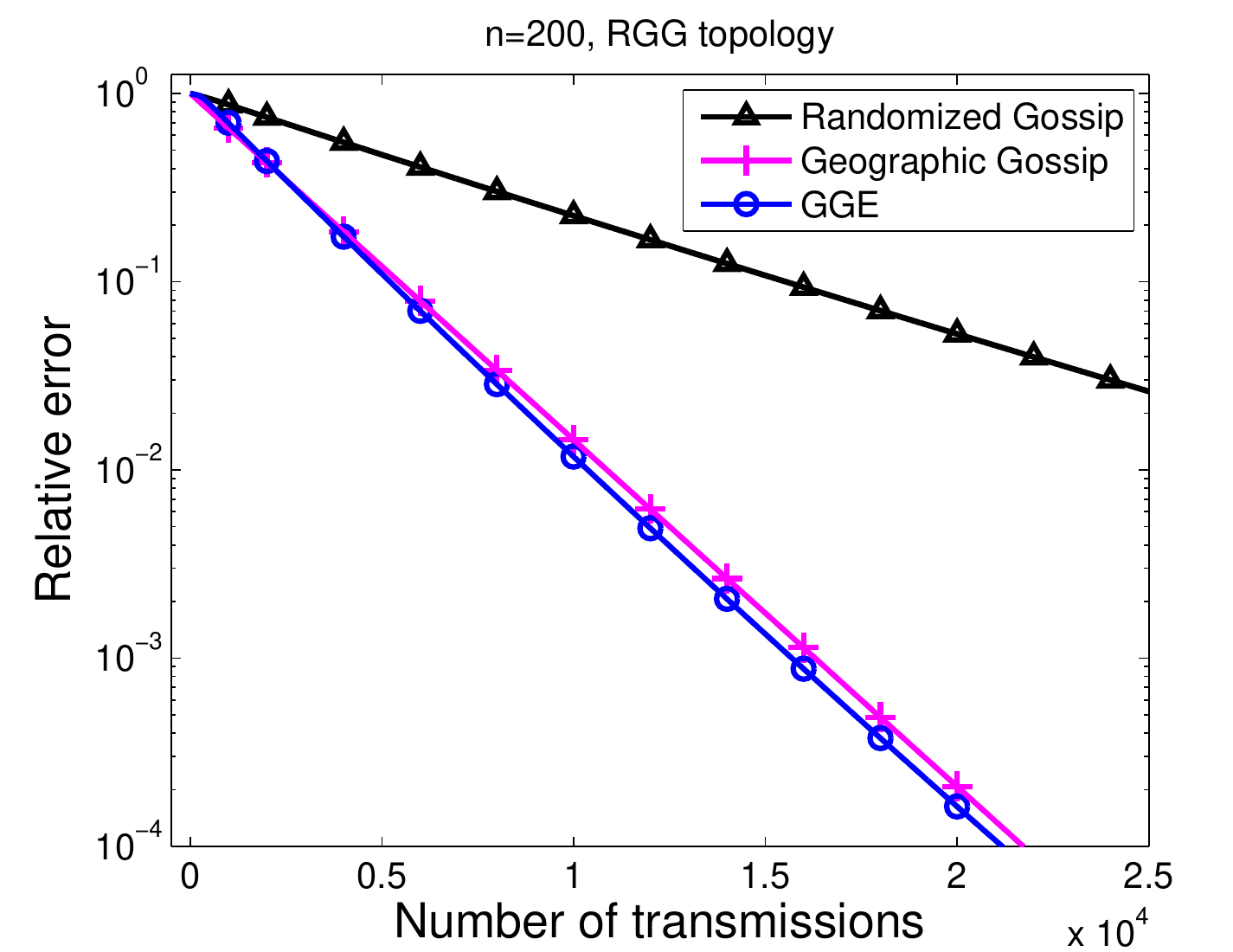}}
\hfill \subfigure[Spike convergence rate comparison]{
\label{fig:convcompc} %% label for second subfigure
\includegraphics[height = 0.35\columnwidth,clip=true, viewport=1.3in 3.3in 7in 7.7in]{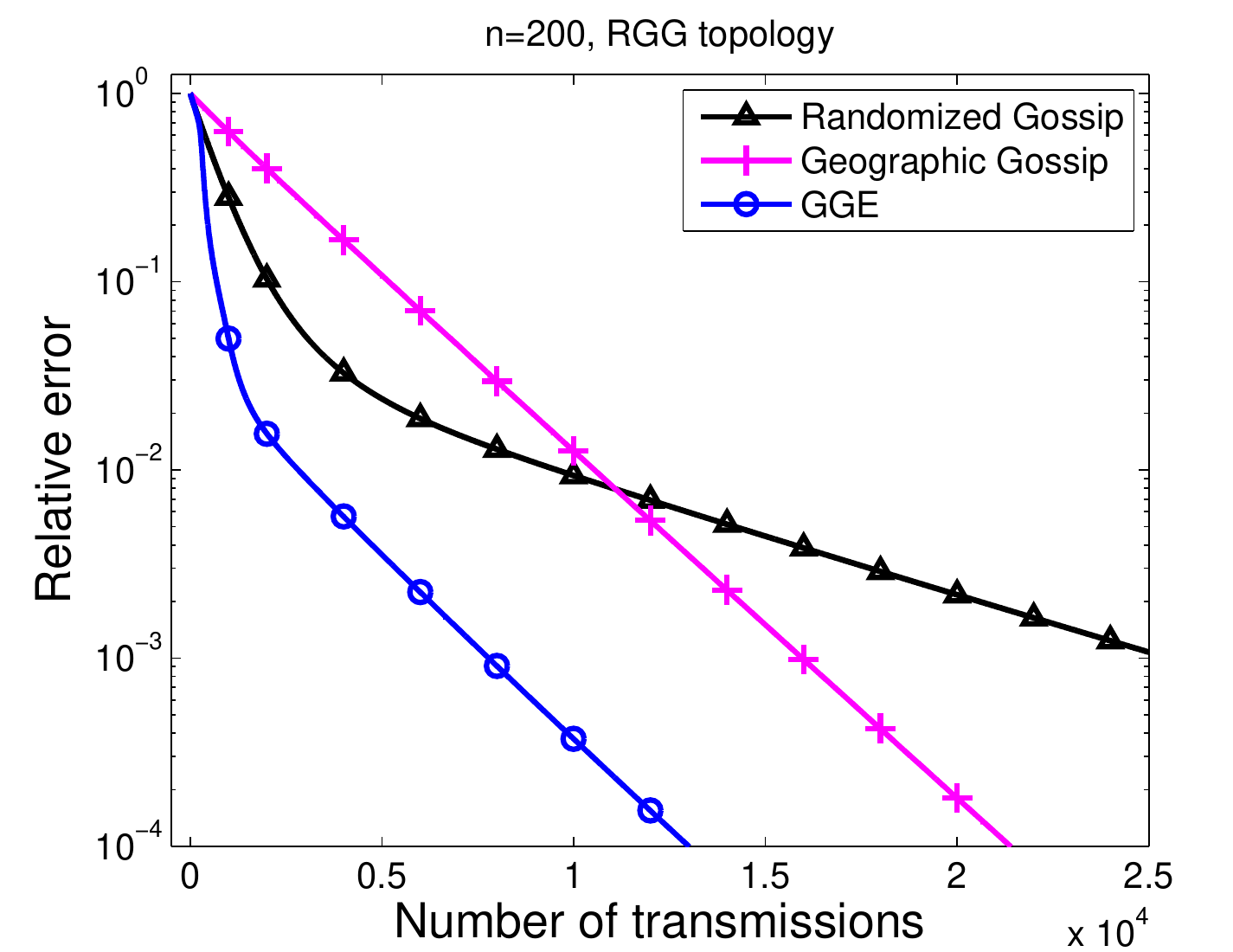}}
\hfill \subfigure[Uniform random field convergence rate comparison]{
\label{fig:convcompd} %% label for second subfigure
\includegraphics[height = 0.35\columnwidth,clip=true, viewport=1.3in 3.3in 7in 7.7in]{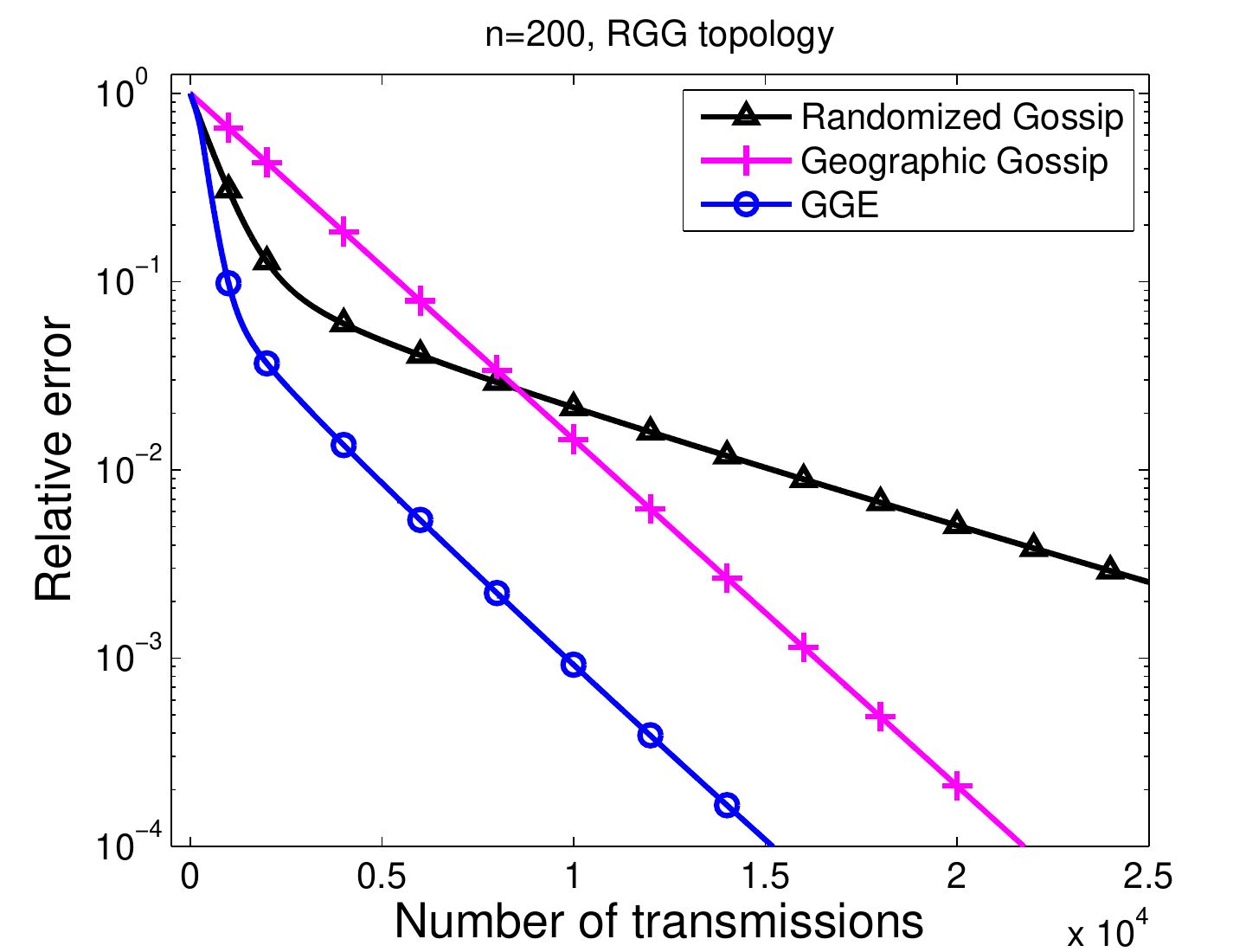}}
\caption{ A comparison of the performance of randomized gossip, GGE, and geographic gossip for four initializations of $x(0)$. Results are averaged over 100 realizations of the random geometric graph and 100 runs of the algorithm per graph. } \label{fig:convcomp}
\end{figure*}

We also compare the performances of the three gossip algorithms for the grid topology. Figure~\ref{fig:convcompGrid} shows that in grid-structured networks the performance of GGE is close to the performance of randomized gossip (constant improvement). Clearly geographic gossip has the best performance in this case. As discussed in Section~\ref{sec:bound}, the small number of neighbors in the grid topology restricts the improvement that GGE can achieve relative to randomized gossip.

\begin{figure*}[tbp]
\centering
\subfigure[Gaussian bumps convergence rate comparison]{
\label{fig:convcompGrida} %% label for second subfigure
\includegraphics[height = 0.35\columnwidth,clip=true, viewport=1.3in 3.3in 7in 7.7in]{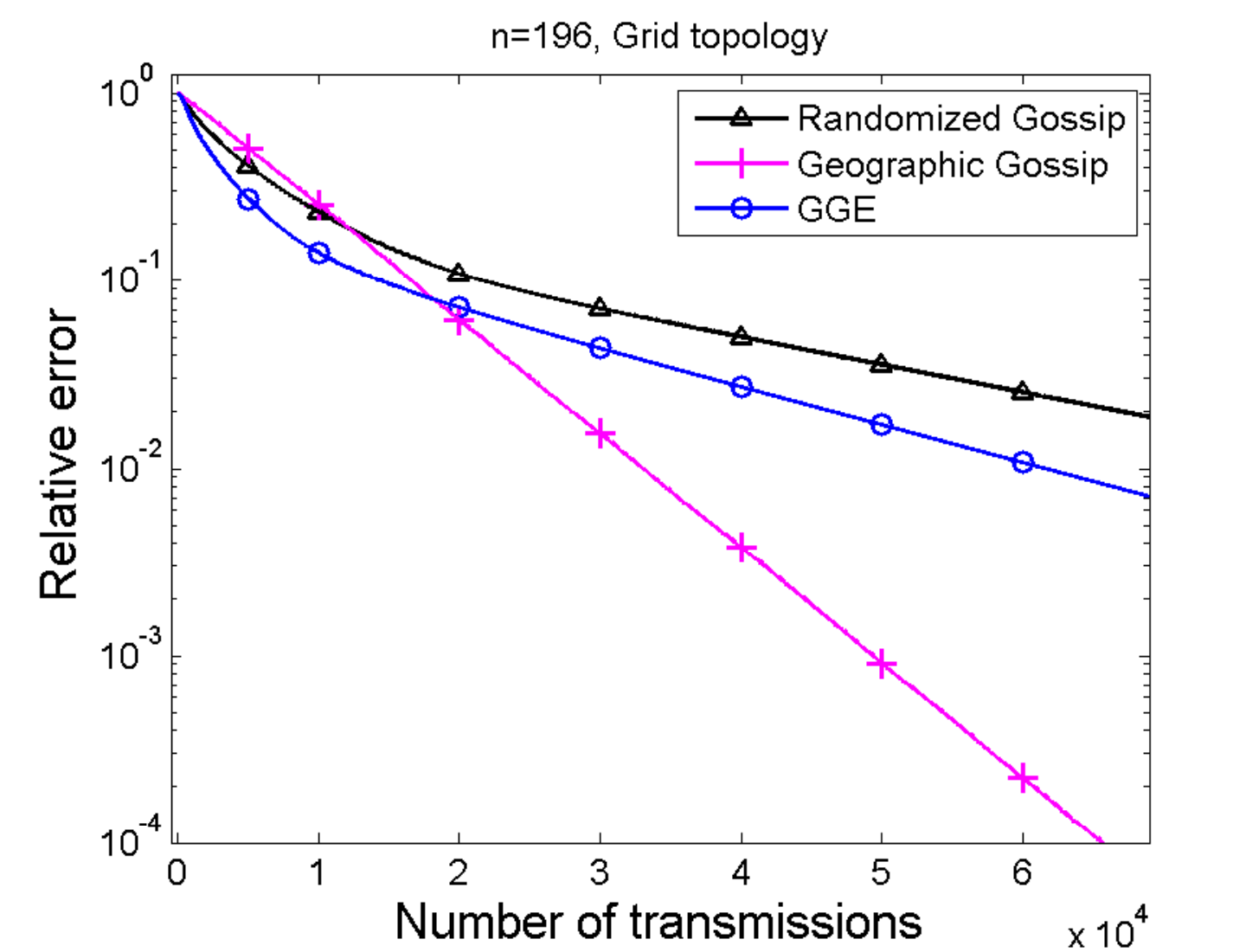}}
\hfill \subfigure[Linearly-varying field convergence rate comparison]{
\label{fig:convcompGridb} %% label for first subfigure
\includegraphics[height = 0.35\columnwidth,clip=true, viewport=1.3in 3.3in 7in 7.7in]{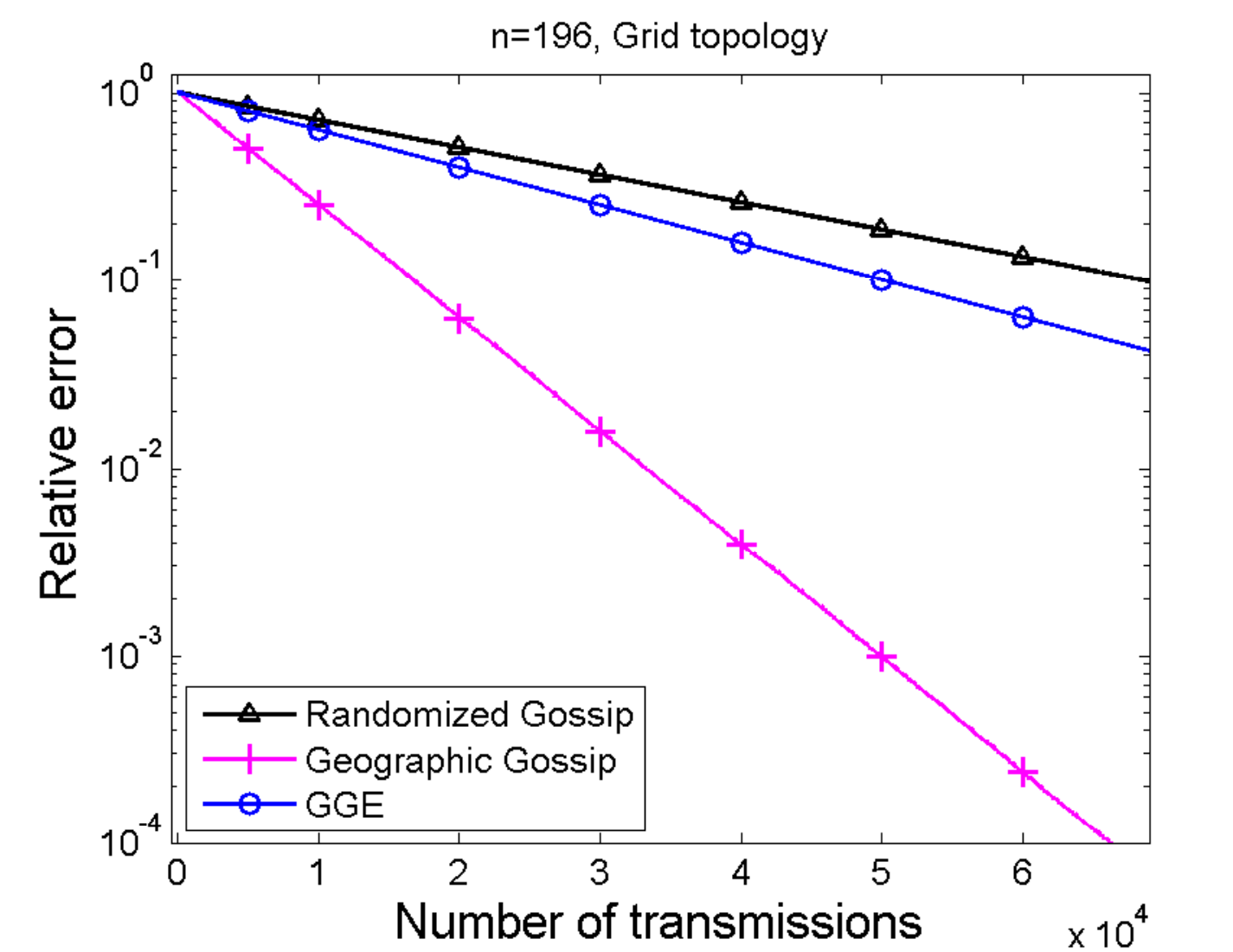}}
\caption{ A comparison of the performance of randomized gossip, GGE, and geographic gossip for two initializations of $x(0)$ in grid topology. Results are averaged over 100 runs of the algorithm. } \label{fig:convcompGrid} %% label for entire figure
\end{figure*}

For random geometric graph topologies, the expected node degree, $\mathbb{E}[|\N_s|]$, scales as $\log n$ in contrast to constant average node degree in the case of grid topology. Therefore, GGE is able to provide better performance results compared to randomized gossip for graph topologies where average node degree increases with the number of nodes.  To improve GGE performance for topologies with low expected node degree values, we propose an extension to the algorithm. Details of this extension, which we call Multi-hop GGE, are provided in Section~\ref{sec:MH}.

\subsection{Comparison with the Theoretical Upper Bound}

We now compare the empirical average relative error for the random geometric graph with the bound developed in Theorem 3. There is no closed-form solution for $A(G)$, so we solve the optimization problem identified in Theorem 3 numerically, using an incremental subgradient algorithm. Since the cost function can be expressed as a function of $(x(k) - \bar x)/||x(k) - \bar x||$, we can focus on maximizing over $x$ satisfying $\bar x = 0$ and $||x(k)||^2 = 1$.  In this simplified setting, one can reformulate the optimization as the minimization of a convex function over a non-convex set of constraints.  We approximate the solution to this minimization using a projected incremental subgradient method.  To avoid the problem of local minima (since the constraint set is non-convex) we rerun the optimization algorithm from multiple initial conditions. Figure~\ref{fig:errorbound} shows the relative error achieved by GGE as a function of the number of iterations for different initial conditions of $x(0)$, averaged over 100 realizations of the algorithm. Also plotted is the bound identified in Theorem 3, substituting in $A(G)$ calculated numerically.  For all but the linearly-varying field, GGE achieves a much more rapid initial decrease in error than indicated by the bound. After approximately 1000 iterations, the bound provides a good indication of the rate of decrease in error.  We again observe that the linearly-varying field is close to a worst-case scenario for GGE.

%\addtolength{\textheight}{-1.0cm}
\begin{figure}[h]
\centering
\includegraphics[height = 0.35\columnwidth,clip=true, viewport=1.3in 3.3in 7in 7.7in]{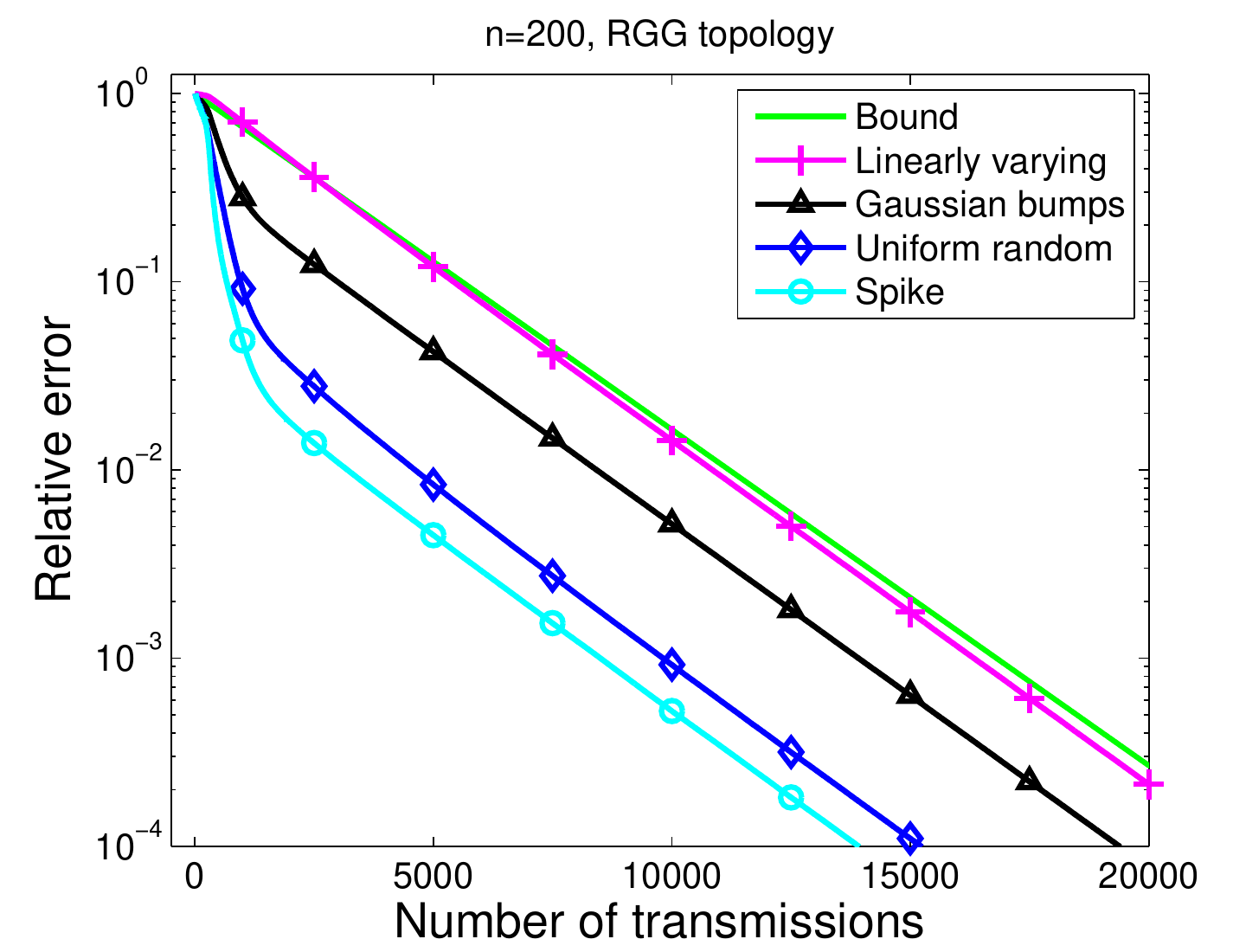}
\caption{
  A comparison of the theoretical bound on relative error and the experimental performance of GGE for four initializations. Results are for 100 realizations of the random geometric graph, averaged over 100 runs of the algorithm.} \label{fig:errorbound} %% label for entire figure
  \end{figure}

Next, we examine how the communication complexity scales with respect to the number of nodes in the network. Figs.~\ref{fig:numnodesRGG} and \ref{fig:numnodesGRID} display how $A(G)$ and the averaging time scale as a function of the number of nodes $n$, for random geometric graph and grid topologies, respectively. To obtain the random geometric graph curve, we generate 50 random graphs for each value of $n$, and numerically evaluate $A(G)$ for each using the procedure detailed above. The top panel shows how the values of $A(G)$ change as the number of nodes increases. The bottom panel shows the $\epsilon$-averaging time, $T_{ave}(\epsilon)$, evaluated via simulation, for $\epsilon=0.01$ versus the number of nodes. Note that Figs.~\ref{fig:numnodesRGG} and \ref{fig:numnodesGRID} show the averaging time in terms of the number of iterations per node. The errorbars depict the minimum, mean and maximum values obtained for the 50 simulated graphs for each $n$. For reference, the dotted lines depict $1.5n/\log{n}$ for the random geometric graph and $2.5{n}$ for the grid topology.

\begin{figure*}[ht]
\centering
\subfigure[RGG topology]{
\label{fig:numnodesRGG} %% label for second subfigure
\includegraphics[height = 0.35\columnwidth,clip=true, viewport=1.3in 3.3in 7in 7.7in]{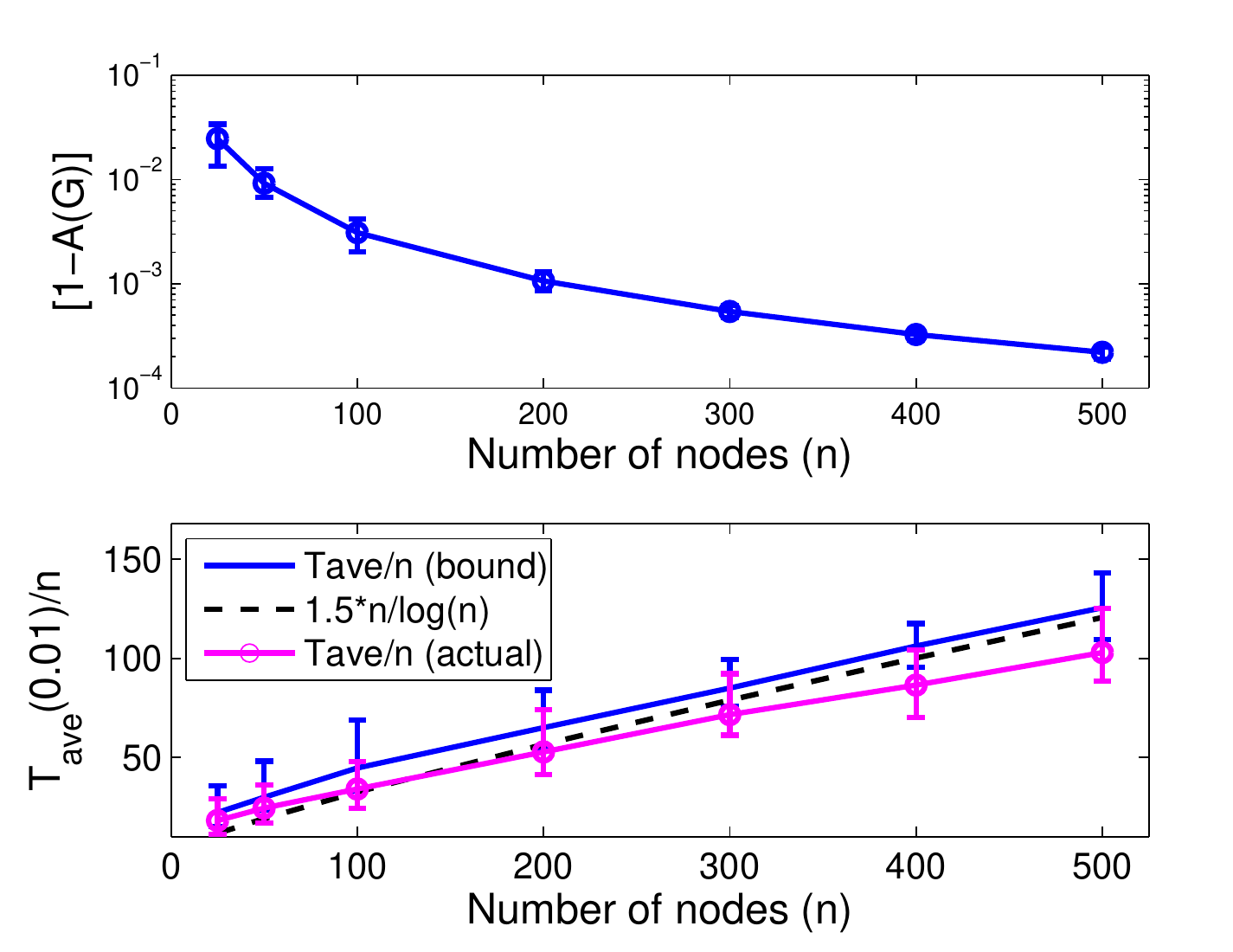}}
\hfill \subfigure[Grid topology]{
\label{fig:numnodesGRID} %% label for first subfigure
\includegraphics[height = 0.35\columnwidth,clip=true, viewport=1.3in 3.3in 7in 7.7in]{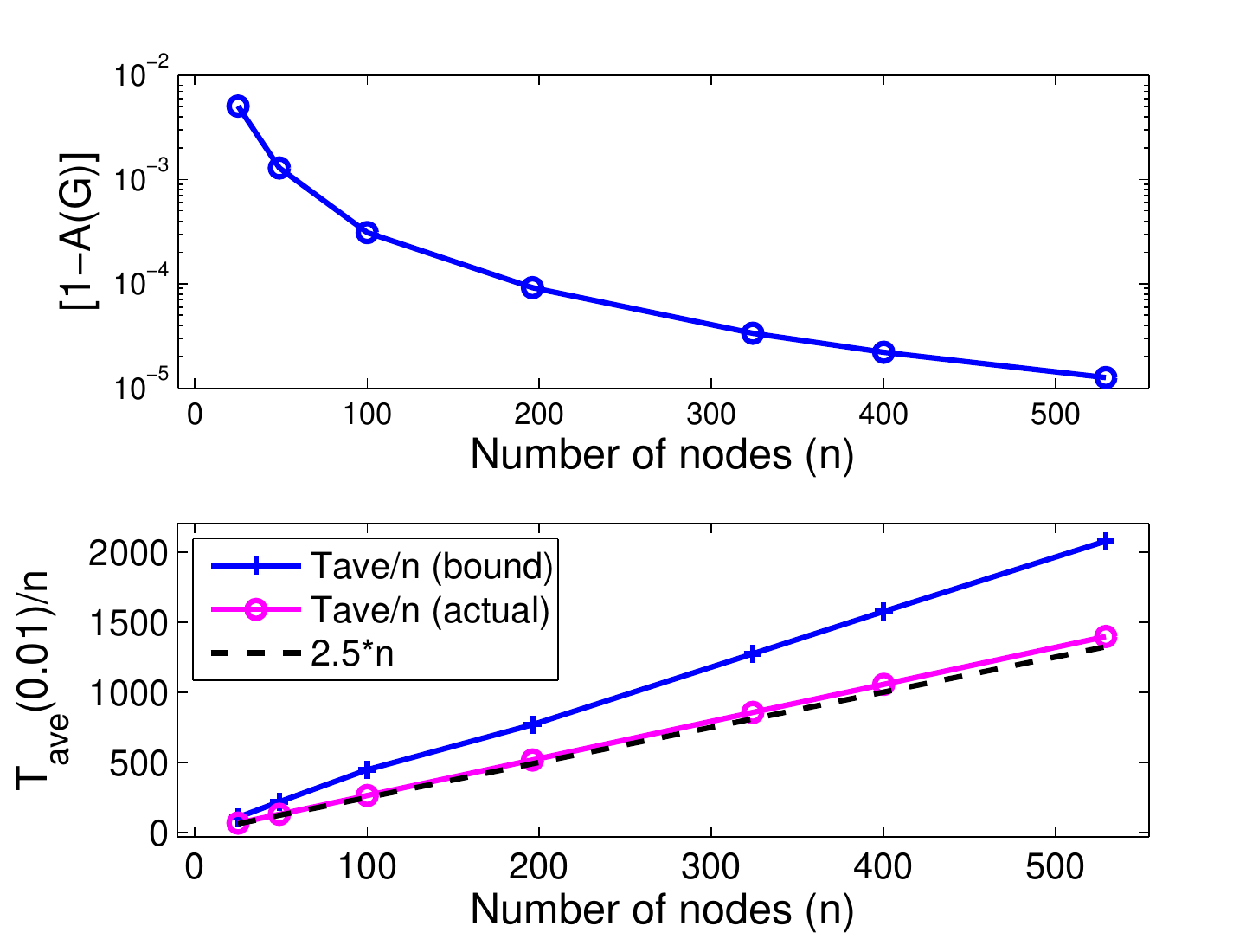}}
\caption{The scaling behavior of $A(G)$ and the averaging time $T_{ave}(\epsilon)$ for $\epsilon=0.01$ as a function of the number of nodes $n$ in the network. The top panels show $A(G)$ evaluated numerically, and the bottom panels compare $T_{ave}(0.01)$ computed via simulation and via the corresponding bound from Theorem~\ref{prop:worstCase}.  (a) 50 random geometric graphs are simulated for each value of $n$. The error bars depict the minimum, mean, and maximum values obtained over these 50 realizations. The curve $1.5n/\log{n}$ is also shown for reference. (b) For the grid topology, error bars are not used since there is only one realization of the grid for a given network size. The curve $2.5n$ is also shown for reference. } \label{fig:numnodes} %% label for entire figure
\end{figure*}

\subsection{Stale Information}  \label{subsec:staleInfo}
In wireless networks, links can be unreliable, and for GGE, it is possible for nodes to miss some updates from their neighbors. Consequently, nodes will have stale information about their neighbors' values and therefore the greedy selection in GGE may be affected. Here we investigate the effect of stale information on the performance of GGE through a simulation study.

We consider random graph topologies with 200 nodes. The initial measurements $x(0)$ correspond to sampling the Gaussian bumps field, similar to Figure~\ref{fig:convcomp}(a).  As described in Section~\ref{sec:alg}, at the $k$th GGE iteration, two nodes $s_k$ and $t_k$ perform averaging. To provide up-to-date information to their neighbors, $s_k$ and $t_k$ broadcast their new values. We simulate the case when nodes randomly miss the broadcasted messages.  We assume that the gossiping nodes $s_k$ and $t_k$ communicate reliably, but sometimes eavesdropping nodes miss an update from their neighbor.  Specifically, each eavesdropping node independently misses the transmission from its neighbor with probability $p$.  Figure~\ref{fig:stale} illustrates the performance degradation in GGE.  Curves are shown for four different values of $p$ between 0 and 0.5, and standard randomized gossip is also shown for comparison.  We conclude that GGE provides significantly better performance than randomized gossip even when 50 percent of the broadcast messages are missed.

\begin{figure}[h] \centering \includegraphics[height = 0.35\columnwidth,clip=true, viewport=1.3in 3.3in 7in 7.7in]{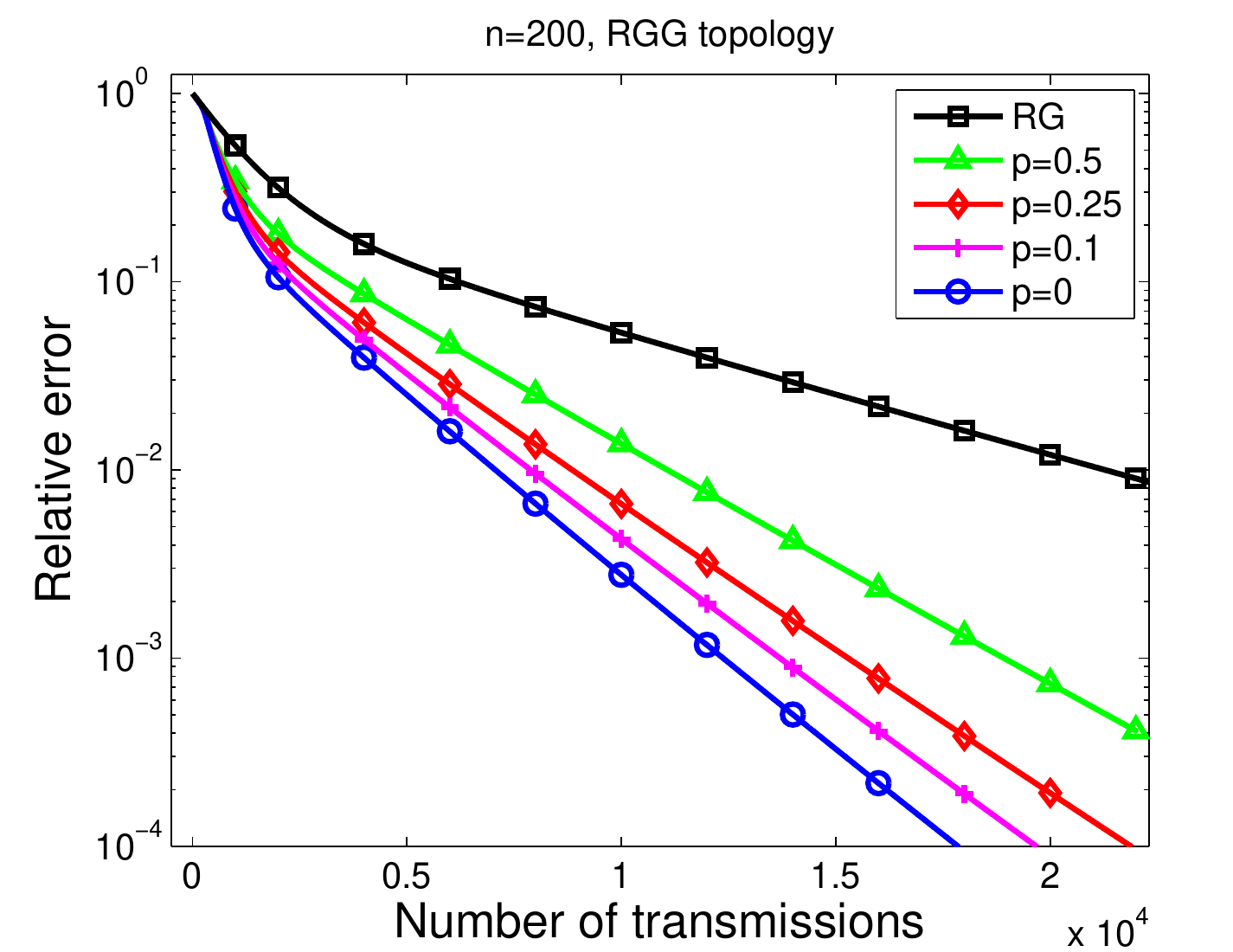} \caption{ A comparision of the performance of randomized gossip with GGE in the case of stale information due to independent link failures with probability $p$. Results are for 100 realizations of the 200-node random geometric graph, averaged over 100 runs of the algorithm. } \label{fig:stale} %% label for entire figure
 \end{figure}

\section{Multi-hop Greedy gossip with eavesdropping}
\label{sec:MH}

As Figs.~\ref{fig:convcomp} and \ref{fig:convcompGrid} indicate, the improvement of GGE over randomized gossip is less for the grid topology compared to the RGG topology. The decrease in improvement is due to the fact that the node degree in a two-dimensional grid is bounded at 4 and does not increase with the network size.  Here we propose an extension to our algorithm that improves the performance of GGE such that it can be employed for topologies with low average node degree. Essentially, this extension allows nodes to perform greedy gossip updates with nodes beyond their immediate one-hop neighborhood.

In one-hop GGE, at the $k$th iteration, $s_k$ determines which neighbor, $t_k$, has a value most different from its own. In two-hop GGE, instead of completing the update, $t_k$ checks if any of its neighbors $u_k\in \N_{t_k}$ has a value even more different from $s_k$ than its own; i.e., $\|x_{s_k}-x_{u_k}\|>\|x_{s_k}-x_{t_k}\|$. If so, nodes $s_k$ and $u_k$ gossip; otherwise $s_k$ and $t_k$ gossip.  Multi-hop gossip generalizes this idea to even larger neighborhoods.  For example, $u_k$ can search its neighborhood, and so on.

To observe the effect of performing greedy updates over multiple hops, we conduct an experimental comparison between 1-hop, 2-hop, and 3-hop GGE.  As a point of comparison, we also include curves for randomized gossip and geographic gossip. Figure~\ref{fig:MHgrid} illustrates the results for grid and random geometric graph topologies. In the grid, 3-hop GGE achieves an asymptotic rate of reduction in relative error that is comparable to geographic gossip. In the random geometric graph topology, the asymptotic performance of 1-hop GGE is already similar to that of geographic gossip (for a network of 200 nodes) and the multi-hop versions lead to significant improvements while still limiting all gossip exchanges to be between nodes separated by at most three hops.

\begin{figure*}[tbp]
\centering
\subfigure[Linearly-varying field convergence rate comparison]{
\label{fig:MHgridslope} %% label for first subfigure
\includegraphics[height = 0.35\columnwidth,clip=true, viewport=1.3in 3.3in 7in 7.7in]{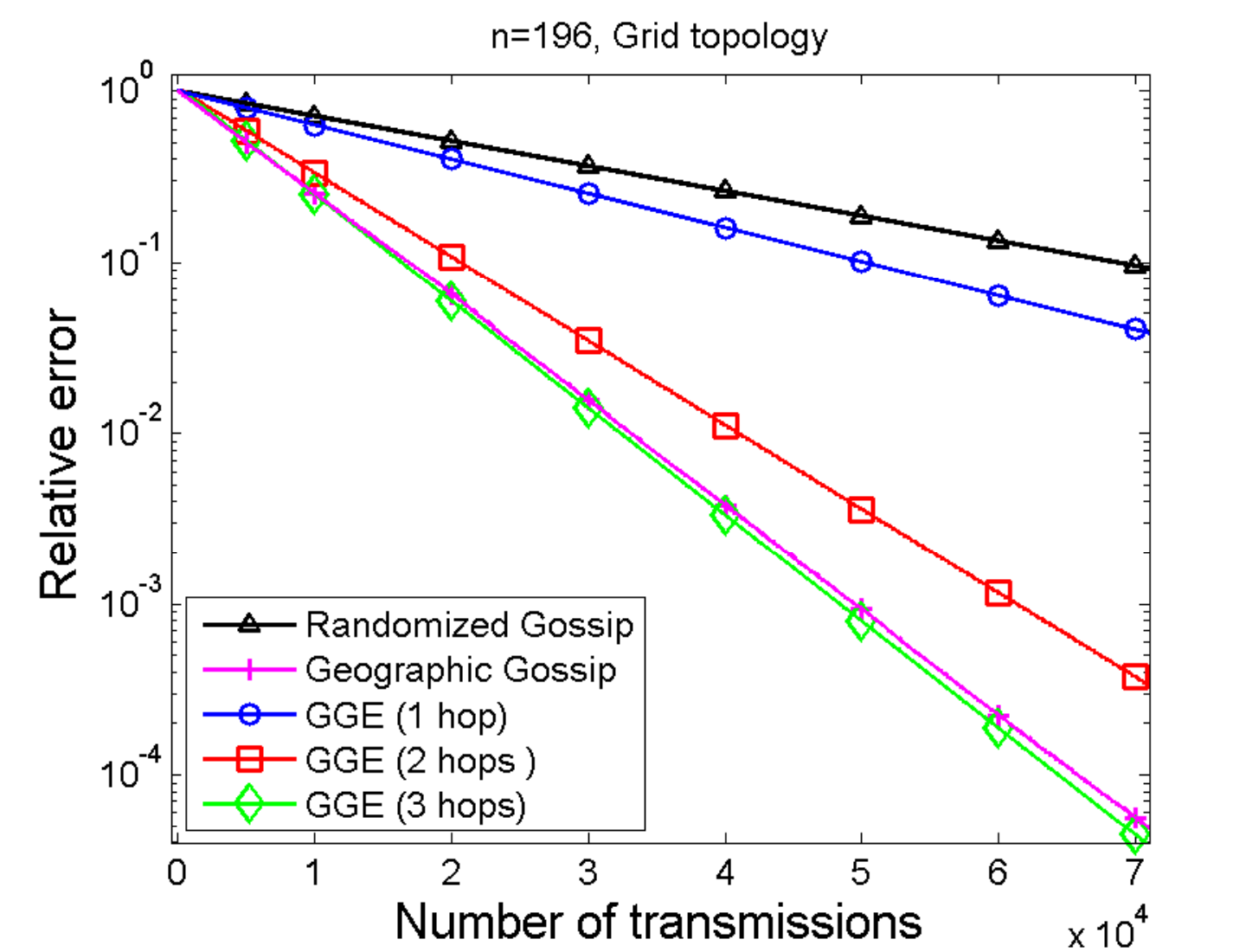}}
\hfill \subfigure[Gaussian bumps convergence rate comparison]{
\label{fig:MHgridslope} %% label for first subfigure
\includegraphics[height = 0.35\columnwidth,clip=true, viewport=1.3in 3.3in 7in 7.7in]{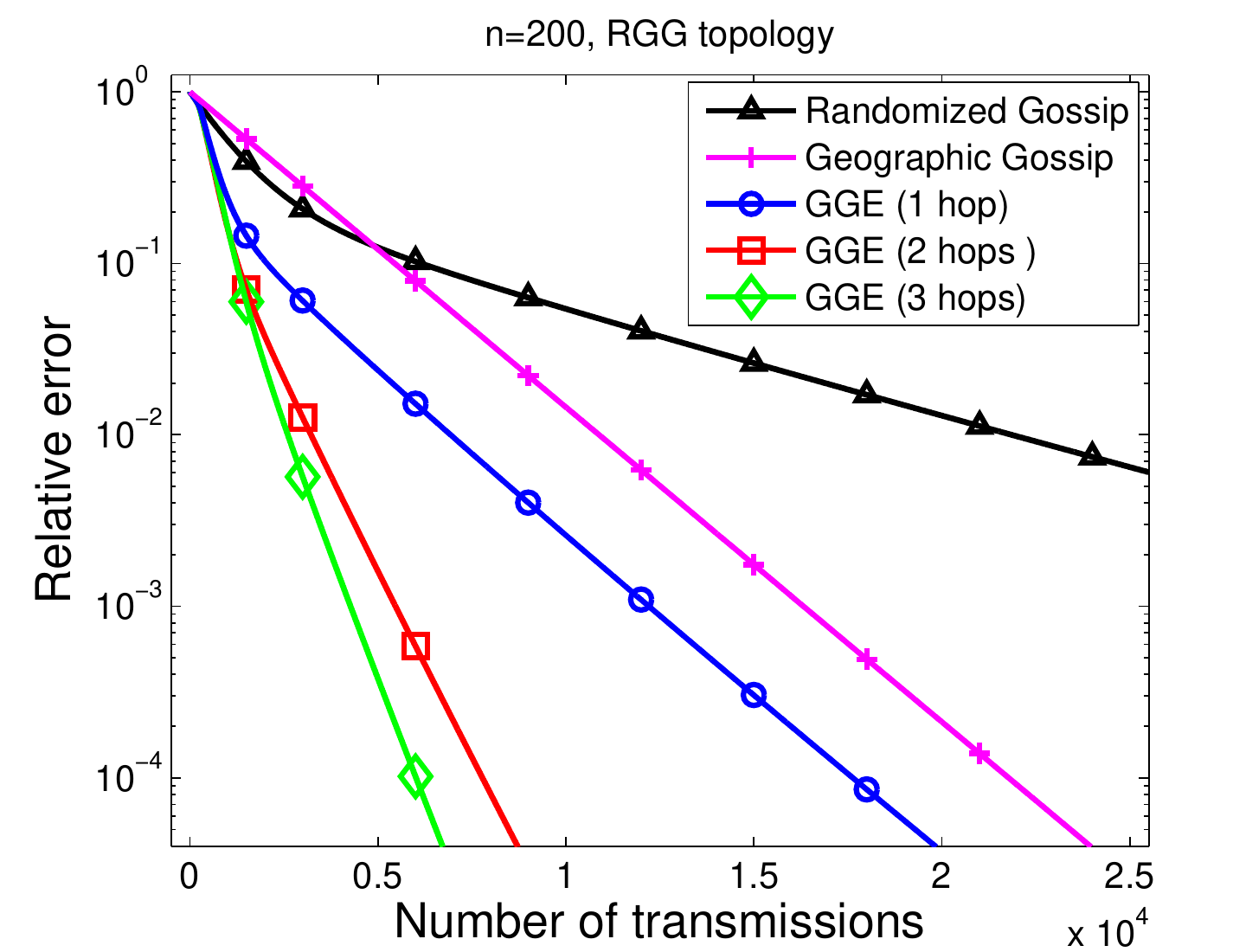}}
\caption{ A comparison of the performance of randomized gossip, GGE (1-hop), multi hop GGE (2 and 3-hops) and geographic gossip for linearly-varying field initialization of $x(0)$ in grid topology and Gaussian bumps initialization of $x(0)$ in RGG topology. Results are averaged over 100 runs of the algorithm.} \label{fig:MHgrid} %% label for entire figure
\end{figure*}

\section{Conclusion}
\label{sec:summ}
In this paper we propose a new average consensus algorithm for wireless sensor networks. Greedy gossip with eavesdropping (GGE) makes use of the broadcast nature of wireless communications and provides fast and reliable computation of average consensus. We provide (i) a proof that GGE converges to the average consensus; (ii) a bound on the mean-squared error after $k$ iterations of  GGE; (iii) a bound on the $\epsilon$-averaging time of GGE; and (iv) theoretical bounds suggesting that GGE converges faster than randomized gossip, and (v) a characterization of the improvement in convergence rate achieved by GGE over randomized gossip as a function of the maximum degree. Simulation experiments compare the performance of GGE, randomized gossip~\cite{boyd06}, and geographic gossip~\cite{dimakis06} and demonstrate that the theoretical bound on mean-squared error provides a good characterization of the algorithm performance.  The simulation experiments also investigate the scaling behavior of the communication complexity of GGE.

GGE retains the robustness and simplicity of randomized gossip; it does not require nodes to acquire location information and it does not introduce the overhead of geographic routing. There is an additional memory overhead (nodes store their neighbors' values), but this storage requirement is small. Nodes do need to learn their neighbors' values, and we propose an initialization process that introduces a minor performance penalty with negligible added complexity.  Since nodes eavesdrop on their neighbors' broadcasts, they must remain in ``Receive" mode throughout the entire operation of the GGE algorithm. In randomized gossip, nodes can enter ``Idle" mode and only switch to ``Receive" mode when they detect that a neighbor is requesting a data exchange. In a wireless sensor network implementation, this difference could lead to concerns that GGE would consume more energy than randomized gossip. However, empirical studies have shown that energy consumption in ``Idle" and ``Receive" modes is very similar for most existing wireless sensor network architectures \cite{RSPS02}.

Our future work will investigate the benefits of GGE in networks of mobile nodes.  When nodes are mobile, other fast consensus approaches which exploit knowledge of geographic location are no longer applicable.  However, because GGE is purely local and adaptive, we believe it is a promising candidate for accelerating gossip algorithms in time-varying networks.  We also plan to investigate further connections between consensus algorithms and incremental subgradient optimization algorithms, towards computing more general functions than the average.

\bibliographystyle{IEEEtran}
\bibliography{gossip}

\begin{thebibliography}{10}
\providecommand{\url}[1]{#1}
\csname url@samestyle\endcsname
\providecommand{\newblock}{\relax}
\providecommand{\bibinfo}[2]{#2}
\providecommand{\BIBentrySTDinterwordspacing}{\spaceskip=0pt\relax}
\providecommand{\BIBentryALTinterwordstretchfactor}{4}
\providecommand{\BIBentryALTinterwordspacing}{\spaceskip=\fontdimen2\font plus
\BIBentryALTinterwordstretchfactor\fontdimen3\font minus
  \fontdimen4\font\relax}
\providecommand{\BIBforeignlanguage}[2]{{%
\expandafter\ifx\csname l@#1\endcsname\relax
\typeout{** WARNING: IEEEtran.bst: No hyphenation pattern has been}%
\typeout{** loaded for the language `#1'. Using the pattern for}%
\typeout{** the default language instead.}%
\else
\language=\csname l@#1\endcsname
\fi
#2}}
\providecommand{\BIBdecl}{\relax}
\BIBdecl

\bibitem{ustebay08}
D.~\"Ustebay, M.~Coates, and M.~Rabbat, ``Greedy gossip with eavesdropping,''
  in \emph{Proc. IEEE Int. Symp. on Wireless Pervasive Computing}, Santorini,
  Greece, May 2008.

\bibitem{ustebay08a}
D.~\"Ustebay, B.~Oreshkin, M.~Coates, and M.~Rabbat, ``Rates of convergence for
  greedy gossip with eavesdropping,'' in \emph{Proc. Allerton Conf. on Comm.,
  Control, and Computing}, IL, USA, September 2008.

\bibitem{ustebay09}
------, ``The speed of greed: Characterizing myopic gossip through network
  voracity,'' to appear in \emph{Proc. IEEE Int. Conf. on Acoustics, Speech,
  and Signal Processing (ICASSP)}, Taipei, Taiwan, April 2009.

\bibitem{tsitsiklis84}
J.~Tsitsiklis, ``Problems in decentralized decision making and computation,''
  Ph.D. dissertation, Massachusetts Institute of Technology, 1984.

\bibitem{SundharRam08a}
S.~{Sundhar~Ram}, V.~Veeravalli, and A.~Nedi\'{c}, ``Distributed and recursive
  parameter estimation in parametrized linear state-space models,'' Submitted,
  Apr. 2008.

\bibitem{rabbat05}
M.~Rabbat, R.~Nowak, and J.~Bucklew, ``Robust decentralized source localization
  via averaging,'' in \emph{Proc. IEEE Int. Conf. on Acoustics, Speech, and
  Signal Processing (ICASSP)}, Phil., PA, Mar. 2005.

\bibitem{rabbat06}
M.~Rabbat, J.~Haupt, A.~Singh, and R.~Nowak, ``Decentralized compression and
  predistribution via randomized gossiping,'' in \emph{Proc.~Information
  Processing in Sensor Networks}, Nashville, TN, Apr. 2006.

\bibitem{kashyap07}
A.~Kashyap, T.~Basar, and R.Srikant, ``Quantized consensus,''
  \emph{Automatica}, vol.~43, pp. 1192--1203, Jul. 2007.

\bibitem{sundaram08}
S.~Sundaram and C.~Hadjicostis, ``Distributed function calculation and
  consensus using linear iterative strategies,'' \emph{IEEE J. Selected Areas
  in Communications}, vol.~26, no.~4, pp. 650--660, May 2008.

\bibitem{dimakis06}
A.~Dimakis, A.~Sarwate, and M.~Wainwright, ``Geographic gossip: Efficient
  aggregation for sensor networks,'' in \emph{Proc. Int. Conf. Inf. Proc. in
  Sensor Networks (IPSN)}, Nashville, TN, Apr. 2006.

\bibitem{xiao04}
L.~Xiao and S.~Boyd, ``Fast linear iterations for distributed averaging,''
  \emph{Systems and Control Letters}, vol.~53, no.~1, pp. 65--78, Sep. 2004.

\bibitem{boyd06}
S.~Boyd, A.~Ghosh, B.~Prabhakar, and D.~Shah, ``Randomized gossip algorithms,''
  \emph{IEEE Trans. Info. Theory}, vol.~52, no.~6, pp. 2508--2530, June 2006.

\bibitem{gupta00}
P.~Gupta and P.~Kumar, ``The capacity of wireless networks,'' \emph{IEEE Trans.
  Info. Theory}, vol.~46, no.~2, pp. 388--404, March 2000.

\bibitem{li07}
W.~Li and H.~Dai, ``Location-aided distributed averaging algorithms:
  Performance lower bounds and cluster-based variant,'' in \emph{Proc. Allerton
  Conf. on Comm., Control, and Computing}, Urbana-Champaign, IL, Sep. 2007.

\bibitem{jung07}
K.~Jung, D.~Shah, and J.~Shin, ``Fast gossip through lifted {M}arkov chains,''
  in \emph{Proc. Allerton Conf. on Comm., Control, and Computing},
  Urbana-Champaign, IL, Sep. 2007.

\bibitem{benezit07}
F.~Benezit, A.~Dimakis, P.~Thiran, and M.~Vetterli, ``Gossip along the way:
  Order-optimal consensus through randomized path averaging,'' in \emph{Proc.
  Allerton Conf. on Comm., Control, and Computing}, Urbana-Champaign, IL, Sep.
  2007.

\bibitem{aysal08}
T.~Aysal, M.~Yildiz, and A.~Scaglione, ``Broadcast gossip algorithms,'' in
  \emph{Proc. IEEE Information Theory Workshop}, Porto, Portugal, May 2008.

\bibitem{aysal08b}
T.~Aysal, M.~Yildiz, A.~Sarwate, and A.~Scaglione, ``Broadcast gossip
  algorithms: Design and analysis for consensus,'' in \emph{Proc. IEEE Conf. on
  Decision and Control}, Cancun, Mexico, Dec. 2008.

\bibitem{SundharRam08b}
S.~{Sundhar~Ram}, A.~Nedi\'{c}, and V.~Veeravalli, ``Incremental stochastic
  subgradient algorithms for convex optimization,'' to appear in \emph{SIAM J.
  on Optimization}.

\bibitem{nedic07}
A.~Nedi\'{c} and A.~Ozdaglar, ``Distributed subgradient methods for multi-agent
  optimization,'' \emph{IEEE Trans. Automatic Control}, vol.~54, no.~1, pp.
  48--61, Jan. 2009.

\bibitem{bertsekasPDC}
D.~Bertsekas and J.~Tsitsiklis, \emph{Parallel and Distributed Computation:
  Numerical Methods}.\hskip 1em plus 0.5em minus 0.4em\relax Belmont, MA:
  Athena Scientific, 1997.

\bibitem{nedic01}
A.~Nedi\'{c} and D.~Bertsekas, ``Incremental subgradient methods for
  nondifferentiable optimization,'' \emph{SIAM J. on Optimization}, vol.~12,
  no.~1, pp. 109--138, 2001.

\bibitem{burnashev74}
M.~Burnashev and K.~Zigangirov, ``An interval estimation problem for controlled
  observations,'' \emph{Problems in Information Transmission}, vol.~10, pp.
  223--231, 1974.

\bibitem{castro07}
R.~Castro and R.~Nowak, ``Active learning and sampling,'' in \emph{Foundations
  and Applications of Sensor Management}, A.~Hero, D.~Castanon, D.~Cochran, and
  K.~Kastella, Eds.\hskip 1em plus 0.5em minus 0.4em\relax Springer-Verlag,
  2007, pp. 177--200.

\bibitem{RSPS02}
S.~P. V.~Raghunathan, C.~Schurgers and M.~B. Srivastava, ``Energy-aware
  wireless microsensor networks,'' \emph{IEEE Signal Processing Magazine},
  vol.~19, no.~2, pp. 40--50, March 2002.

\end{thebibliography}

\end{document}